\documentclass[letterpaper,12pt,reqno]{article}

\usepackage[letterpaper, left=2cm,right=2cm,top=2cm,bottom=2cm, includefoot,heightrounded]{geometry}
\usepackage{setspace}

\usepackage{mathpazo} 

\usepackage{xurl}
\usepackage{hyperref}
\hypersetup{colorlinks=true, linkcolor=blue, urlcolor  = blue,  citecolor=blue}
\usepackage{titlesec}
\usepackage[utf8]{inputenc}
\usepackage{graphicx}
\usepackage{amsmath}
\usepackage{MnSymbol}%
\usepackage{wasysym}%
\graphicspath{{Images/}}
\usepackage{amsfonts}
\usepackage{algpseudocode}
\usepackage{enumitem} 
\usepackage{xcolor}
\usepackage[english]{babel}
\usepackage{amsthm}
\usepackage{tocloft}
\usepackage{bbm}
\usepackage{mathtools}
\usepackage{tikz}
\usepackage{comment}
\usetikzlibrary{decorations.pathreplacing}
\usetikzlibrary{positioning,quotes,calc}
\usepackage{subcaption}
\newtheorem*{theorem*}{Theorem}
\newtheorem*{corollary*}{Corollary}
\newtheorem*{proposition*}{Proposition}
\newtheorem*{lemma*}{Lemma}
\newtheorem*{fact*}{Fact}
\newtheorem*{definition*}{Definition}
\newtheorem*{conjecture*}{Conjecture}
\newtheorem{theorem}{Theorem}

\newtheorem{proposition}{Proposition}
\newtheorem{assumption}{Assumption}

\newtheorem{definition}{Definition}
\newtheorem{lemma}{Lemma}
\newtheorem{remark}{Remark}
\newtheorem{fact}{Fact}

\DeclarePairedDelimiterX{\inp}[2]{\langle}{\rangle}{#1, #2}

\newcommand{\overbar}[1]{\mkern 1.5mu\overline{\mkern-1.5mu#1\mkern-1.5mu}\mkern 1.5mu}

\titleformat{\section}
		{\large\bfseries\center}
         {\thesection}
        {0.5em}
        {}
        []
\titleformat{\subsection}
        {\normalfont\bfseries}
        {\thesubsection}
        {0.5em}
        {}
\titleformat{\subsubsection}[runin]
        {\normalfont\bfseries}
        {\thesubsubsection}
        {0.5em}
        {}
        [.]

        \usepackage{titling}

\usepackage[authoryear]{natbib}
\usepackage{soul}
\title{Equitable screening\thanks{I am grateful to Piotr Dworczak, Joey Feffer, Ben Golub, Ravi Jagadeesan, Annie Liang, Federico Llarena, Paula Onuchic, Sam Wycherley, and Frank Yang for their helpful comments.}}
\author{Filip Tokarski \\ Stanford GSB}
\begin{document}
\maketitle
\date{\today} 
\vspace*{-1cm} 

\begin{abstract}

A designer distributes goods while considering the perceived equity of the resulting allocation. Such concerns are modeled through an equity constraint requiring that equally deserving agents receive equal allocations. I ask what forms of screening are compatible with equity and show that while the designer cannot equitably screen with a single instrument (e.g., payments or ordeals), combining multiple instruments, which on their own favor different groups, allows her to screen while still producing an equitable allocation.
        \end{abstract}

\section{Introduction}

In 2017, the French president Emmanuel Macron unveiled a set of environmental policies aimed at reducing carbon emissions. The plan involved a gradual increase in fuel taxes, including a significant hike in diesel and petrol prices the following year. However, the proposed tax increase sparked protests and riots, which later expanded into the “Yellow Vest” movement. The protesters, largely from rural and less affluent areas, objected to the policy on equity grounds. They claimed that households in poor financial standing were disproportionately bearing the cost of decarbonization, and that the tax unfairly burdened those who could not cut down on driving or switch to greener alternatives.\footnote{\url{https://www.nytimes.com/2018/12/06/world/europe/france-fuel-carbon-tax.html}} Ultimately, mounting unrest led the French government to suspend the planned fuel tax hike. Similarly motivated opposition derailed green reforms in \href{https://news.climate.columbia.edu/2021/07/30/co\%E2\%82\%82-reduction-law-rejected-in-swiss-referendum}{Switzerland}, \href{https://grist.org/article/the-rural-anxiety-behind-oregons-failed-climate-legislation/}{Oregon}, and \href{https://taxpolicycenter.org/taxvox/why-unusual-coalition-defeated-washington-states-carbon-tax-initiative}{Washington state}.

This paper asks how such perceived equity concerns restrict policies for allocating goods like medical treatments, as well as schemes levying charges for emissions or entering congested city centers. I model equity concerns using a \emph{merit function} specifying societal perceptions of how entitled each agent is to the allocated good. For instance, in the case of emissions, rural households reliant on cars for transportation would be more deserving of a right to emit than urban ones. I then introduce an equity constraint requiring that agents with equal merit receive the same allocation. 

Importantly, the policy designer cannot observe or condition on all the factors that determine a particular agent's merit. In the case of the fuel tax, for instance, the government does not observe one's commute length and frequency and would find it difficult to condition the tax or rebate on one's access to public transport. Nevertheless, the designer is still judged on the fairness of the aggregate allocation: despite each individual's needs and characteristics being difficult to observe, it might still be apparent that certain groups, e.g., poor rural households, are treated unfairly. Thus, to satisfy the equity constraint, the designer will have to anticipate and correctly account for the choices different types of agents make in the mechanism.\footnote{\label{footnote:2} Even if the designer observes some of the applicants' characteristics, such as their income bracket, the problem of selection on unobservables persists within each observable group. Moreover, heterogeneity in characteristics relevant for merit remains substantial even after conditioning on these variables. For instance, current income, a metric on which eligibility is frequently based, is far less predictive of future income, race, and education for the lowest earners than it is for the population at large \citep{cook2023build}.}

The equity constraint will limit the ways in which the designer can screen agents. This is important, as screening is a crucial feature of many public policy designs. For instance, in the carbon tax case, raising the price of driving is meant to discourage (screen out) those with relatively low need for it to reduce emissions. A similar effect can be accomplished by screening agents with non-monetary instruments, such as requiring them to wait in line or fill out forms. Indeed, if participants need to go through such ordeals to get the benefit, only those who need it will do so \citep{zeckhauser}. One might worry, however, that screening could produce inequitable allocations, as people with the same need for the good may find completing ordeals or paying burdensome to different extents. I therefore ask how (if at all) the designer can screen agents when allocations are subject to equity constraints. I look at screening using only payments (which are less costly to the rich), only ordeals (less costly to the poor) and both of these instruments at once. In the former two cases equitable screening is impossible. However, when the designer uses both payments and ordeals, she has significant freedom to screen despite equity constraints. Such mechanisms offer recipients a menu of “payment options”---to get the good, participants can pay a fee, complete an ordeal, or do a combination of the two. 

My results therefore suggest that the equity of payment-based mechanisms, such as congestion or emission charges, could be improved by e.g. pairing them with a (possibly arduous) rebate procedure. In such cases, poorer and more car-reliant individuals could avoid burdensome charges by submitting an exemption application. Such policies are widely used in practice. For instance, Illinois Tollway offers drivers a menu of payments and ordeals: they can pay standard cash tolls or open an account to pay discounted tolls. Households earning below 250\% of the federal poverty line can apply for further reductions: by documenting low income or participation in programs such as SNAP, they can enroll in a special program which waives a deposit and lowers the required amount of prepaid tolls. Such alternative pathways are meant to ``mitigate the disparate impact that transportation costs have on low-income households."\footnote{\url{https://agency.illinoistollway.com/assist}} Menus of payments and ordeals are also common in healthcare. For example, in Australia, patients with non-life-threatening emergencies can choose between waiting at a free public hospital emergency department and paying a fee at a walk-in urgent care clinic to be seen rapidly.\footnote{\url{www.emergencyurgentcare.com.au/services/emergency-clinic/}}

This paper relates to work examining how moral sentiments constrain market designers and policymakers. \cite{repugnance} discusses how the repugnance of certain transactions precludes the use of markets in settings where they would be efficient. The literature following \cite{justenvy} models fairness concerns in matching markets through assigning \emph{priorities} to agents. It then studies matching mechanisms that eliminate \emph{justified envy}---a notion capturing perceived injustice. Finally, \cite{frankel2023testoptional} argue that many US colleges switched to test-optional admissions to reduce public scrutiny of their admission decisions. However, no existing work studies equity constraints in mechanism design problems of the sort this paper considers. In doing so, my paper also relates to the literatures on algorithmic fairness in computer science and on discrimination in economics, which attempt to conceptualize bias, unfairness and discrimination (see \cite{alves2023survey} and \cite{paula} for respective surveys, and Subsection \ref{ref:eqconstsec} for a discussion of related fairness notions). However, both of these literatures focus on problems of classification or statistical inference and typically do not account for the strategic behavior of agents. By contrast, the purpose of this paper is to study mechanisms that are fair after accounting for strategic responses. 

My work is also close to the literature on inequality-aware market design \citep{condorelli, dworczak2021redistribution, akbarpour2020redistributive, vaccines, pai2022taxing}. Following \cite{weitzman}, this line of research recognizes that, in certain settings, willingness-to-pay may not accurately reflect need as people differ in their value for money. It then asks how to design welfare-maximizing mechanisms in such cases. While my paper also features a wedge between need and willingness-to-pay, its focus is different: instead of asking about the welfare-maximizing mechanism, I ask how public perceptions of fairness constrain the designer. The difference between these questions is apparent when considering, for example, the problem of allocating COVID-19 vaccines studied by \cite{vaccines}; the paper shows that the optimal policy combines priorities to vulnerable groups with a market mechanism under which one can pay to be vaccinated early. While optimal from a welfare standpoint, the authors themselves acknowledge that their proposed policy might provoke backlash on fairness grounds. Indeed, public perceptions of fairness need not be aligned with welfare-maximization, and thus can impose additional constraints on the policymaker.

The rest of the paper is structured as follows. I first introduce the model of good provision with an equity constraint. Section \ref{sec:imp} then discusses the forms of screening that are feasible in this environment. In Section \ref{sec:obswealth} I ask what allocation rules can be equitably implemented if the designer also observes agents' wealth. Section \ref{sec:relax} discusses measures of violation of the equity constraint while Section \ref{sec:non-linear} considers a richer specification of agents' utility functions; Section \ref{sec:conc} concludes.

\section{Model}

The designer allocates goods $x\in [0,1]$ to agents with types $(\alpha,\beta)\in \Theta \subset \mathbb{R}\times \mathbb{R}_+$. $\Theta$ is open, convex and bounded. $\beta$ represents an agent's value for the good and $\alpha$ represents her value for money (higher $\alpha$ means the agent is poorer). I consider two screening instruments---payments and ordeals. Payments, $p \in \mathbb{R}$, are more burdensome for poorer agents (higher $\alpha$), while ordeals, $q\in \mathbb{R}_+$, are costlier for richer agents. Utility is given by:
\[
U[\alpha,\beta;x,p,q]= \beta x- w(\alpha) p - z(\alpha) q,
\]
where $w,z:\mathbb{R}\to\mathbb{R}$ are strictly positive and twice continuously differentiable, with 
$w'(\alpha)>0,$ $z'(\alpha)<0$.\footnote{
The allocation of the good $x$ may be interpreted as a probability of getting it. While the linearity of the utility function in payments and ordeals may be a good approximation for cases where these are small, I later consider a more general utility specification in Section \ref{sec:non-linear}.} I assume that need for money $\alpha$ and need for the good $\beta$ are private information (I relax this assumption in Section \ref{sec:obswealth}). The designer therefore chooses an allocation rule $x:\Theta \to [0,1]$, payment rule $p:\Theta \to \mathbb{R}$ and ordeal rule $q:\Theta \to \mathbb{R}_+$ subject to \eqref{eq:IC} and \eqref{eq:IR} constraints:
        \begin{equation}\label{eq:IC}
                \text{for all }  (\alpha,\beta) \in \Theta, \quad U[\alpha,\beta;(x,p,q)(\alpha,\beta)] \geq \sup_{(\alpha',\beta')\in \Theta} U[\alpha,\beta;(x,p,q)(\alpha',\beta')],\tag{IC}
        \end{equation}
        \begin{equation}\label{eq:IR}
                \text{for all }  (\alpha,\beta) \in \Theta, \quad U[\alpha,\beta;(x,p,q)(\alpha,\beta)] \geq 0.\tag{IR}
                \end{equation}
An allocation rule $x^*(\alpha,\beta)$ is \emph{implementable} if there exist payment and ordeal rules $p^*(\alpha,\beta)$ and $q^*(\alpha,\beta)$ such that $(x^*,p^*,q^*)$ satisfies \eqref{eq:IC} and \eqref{eq:IR}. 

\begin{remark}
My results require ordeals to be relatively less costly to the poor than to the wealthy. Indeed, as I discuss later, Theorem \ref{th:1} asserting that combining payments and ordeals allows for rich equitable screening would fail if the wealthy found both payments and ordeals less costly. Nevertheless, \cite{dupas} provide evidence supporting my assumption for some ordeals.\footnote{ The authors compare mechanisms for allocating a benefit that screen agents with prices and by requiring them to travel to a distant office. The selection patterns they observe “provide evidence that time and money costs have different selection properties, (...) consistent with a model in which richer households have a higher value of time.”} There are, however, forms of non-monetary screening for which this assumption is implausible. For instance, \cite{Bettinger}, who study college financial aid applications, provide evidence that highly complex bureaucratic procedures disproportionately discourage low-income applicants.\footnote{The authors note that ``To determine eligibility, students and their families must fill out an eight-page, detailed application called the Free Application for Federal Student Aid (FAFSA), which
has over 100 questions."}
\end{remark}

\subsection{Equity constraint}\label{ref:eqconstsec}

The designer also faces an \emph{equity constraint}, which I model using an exogenous \emph{merit function} $\eta: \mathbb{R}^2 \to \mathbb{R}$ specifying how entitled each agent is to the good. I assume that $\eta(\alpha,\beta)$ is twice continuously differentiable with $\eta_\alpha,\eta_\beta$ strictly positive and uniformly bounded away from zero. That is, agents are more entitled to the good if they value it more or if they are poorer (as richer agents can more easily satisfy their needs without government assistance). Since $\Theta$ is open, bounded and connected, and the merit function is strictly increasing and continuous, the set of values attained by $\eta(\alpha,\beta)$ on $\Theta$ is an open interval, which I denote by $(\underline \eta,\overbar \eta)$.

The equity constraint requires that all agents with the same merit receive equal amounts of the good:
\begin{equation}\label{const:equity}
\eta(\alpha^a,\beta^a) = \eta(\alpha^b,\beta^b) \quad \Longrightarrow \quad x(\alpha^a,\beta^a) = x(\alpha^b,\beta^b).
\tag{E}
\end{equation}
I call an allocation rule satisfying \eqref{const:equity} \emph{equitable}. Note that an allocation rule $x(\alpha,\beta)$ is equitable if and only if it can be written in the form $x(\alpha,\beta)\equiv \hat x(\eta(\alpha,\beta))$ for some $\hat x:(\underline \eta,\overbar \eta)\to [0,1]$.

While the equity constraint does not require that agents with higher merit receive more of the good, this will still be the case for any implementable equitable allocation:
\begin{lemma}\label{lem:inc}
If the allocation rule $x(\alpha,\beta)\equiv \hat x(\eta(\alpha,\beta))$ is implementable, then $\hat x$ is weakly increasing.
\end{lemma}
\begin{proof}
Any implementable $x(\alpha,\beta)$ has to be implementable for the subset of agents with need for money $\alpha=\alpha^a$, for any $\alpha^a$. We can write the utility of such agents as $\beta \cdot x(\alpha^a,\beta) - t(\beta)$, where $t(\beta):= w(\alpha^a) \cdot p(\alpha^a,\beta) +z(\alpha^a) \cdot q(\alpha^a,\beta)$. This is a one-dimensional quasi-linear screening problem so any implementable allocation has to be weakly increasing in the value for the good $\beta$. Hence, $x(\alpha,\beta)$ must be weakly increasing in $\beta$ for every $\alpha$. Since the allocation rule takes the form $\hat x(\eta (\alpha,\beta))$ and the merit function $\eta(\alpha,\beta)$ is strictly increasing in $\beta$, it follows that $\hat x$ has to be increasing.
\end{proof}

\begin{remark}
A mechanism with a constant allocation rule $x(\alpha,\beta)$ is always equitable. If the allocation $x$ is interpreted as a probability of getting the good, this mechanism represents a lottery where all participants have equal chances of receiving it. Such lotteries have been used i.a. for military drafts and school choice because of their perceived fairness \citep{peyton}. 
\end{remark}

It is worth discussing the modeling choices behind my formulation of the equity constraint. First, it is imposed on agents' allocations of the good rather than their utilities. If the constraint were imposed on utilities instead, even minimal heterogeneity in preferences would preclude the existence of equitable and implementable mechanisms allocating the good. This is because agents who value the good more or who are less burdened by payments or ordeals would receive information rents from any such mechanism, meaning that utilities could only be equalized when nothing was allocated. Second, the constraint regulates the allocations of the good, but not the assignments of payments and ordeals. The reason for this is twofold: first, requiring the equality of all three could only produce trivial results. Moreover, who gets access to the good may be more salient to the public than how they ``pay'' for it. For instance, poor people not being able to afford basic medical care is likely to attract more scrutiny than the fact that some pay to get it while others have to wait in line.

My equity constraint is also closely related to the algorithmic fairness literature, which attempts to conceptualize bias and discrimination in the problem of classifying agents based on observed characteristics (see \cite{barocas-hardt-narayanan} for an introduction).

A notion related to mine was proposed by \cite{dwork2012fairness}, who formalize the principle that similar individuals should be treated similarly. They consider the problem of assigning distributions over outcomes to individuals with different characteristics and model fairness through a “Lipschitz condition” that limits how quickly the assigned distribution of outcomes can change as the distance between characteristics grows. My notion, by contrast, concerns only the \emph{equal} treatment of \emph{equally} deserving agents.

My approach is also close to the family of fairness criteria known as \emph{conditional independence}, requiring that some aspect of classification be conditionally independent of \emph{protected characteristics}. If we treat wealth $\alpha$ as a protected characteristic, my equity constraint resembles a criterion from the independence family known as \emph{conditional statistical parity}, which requires that, conditional on certain covariates deemed legitimate, the classification rate be the same across protected groups \citep{Corbett-Davies}. However, my constraint differs from conditional statistical parity in two key respects. First, it does not condition directly on agents' characteristics but on merit. Second, in my model, conditioning on wealth and merit fully pins down agents' types and thus independence reduces to equality of allocations. Requiring allocations to be independent conditional on merit would be a natural extension if the model was enriched with characteristics irrelevant to merit. After such an extension, my approach would nest many notions from the independence family: conditional statistical parity could be recovered by assigning distinct values to each legitimate covariate combination; group fairness, requiring that the allocation not depend on protected characteristics, could be recovered by making the merit function constant.

Finally, \cite{hardt2016equality} propose a related notion of \emph{equal opportunity}, which demands that odds of receiving the good do not depend on group membership \emph{among deserving individuals}. Like in my paper, this notion of fairness puts the burden of discovering whether a specific individual is deserving on the designer. Since \cite{hardt2016equality} focus on a classification problem, that burden is one of data collection. In my setting, by contrast, such information needs to be elicited through designing an incentive-compatible mechanism. 

\section{Implementable equitable allocations}\label{sec:imp}

I now ask what forms of screening are compatible with equity. I study the sets of implementable
equitable allocations in three cases: when the designer uses only payments to screen, when she
uses only ordeals, and when she uses both screening instruments.

\subsection{Screening with payments}\label{screeningpayments}
I first show that the designer cannot equitably screen using only payments.
\begin{proposition} \label{th:1}
Suppose we do not use ordeals to screen, so $q \equiv 0$. Then every equitable and implementable allocation rule $x(\alpha,\beta)$ is constant.
\end{proposition}

\begin{proof}
By Lemma \ref{lem:inc}, any implementable and equitable allocation rule $x(\alpha,\beta)$ can be written as $\hat x(\eta(\alpha, \beta))$, where $\hat x$ is weakly increasing. Note also that the payment rule has to take the form $p(\alpha,\beta)\equiv \hat p(\eta(\alpha,\beta))$ because identical allocations of $x$ must require identical payments. An argument analogous to the proof of Lemma \ref{lem:inc} tells us that any implementable payment rule $p(\alpha,\beta)$ has to be weakly decreasing in $\alpha$. Since the merit function $\eta(\alpha,\beta)$ is strictly increasing in need for money $\alpha$, it follows that $\hat p$ must be weakly decreasing. However, $\hat p$ also has to be weakly increasing---otherwise one could deviate and receive a weakly greater allocation of $x$ for a strictly smaller payment. Therefore, $\hat p$ has to be constant. Such a payment rule can only support a constant allocation.
        \end{proof}
Intuitively, equity requires that poorer and richer agents of equal merit receive the same allocation, even though the richer agents have higher need for the good. However, these richer agents with higher need for the good have greater willingness to pay for it (Figure \ref{fig:11}). Therefore, any mechanism that sells the good will allocate more of it to the richer agents, and hence violate equity.

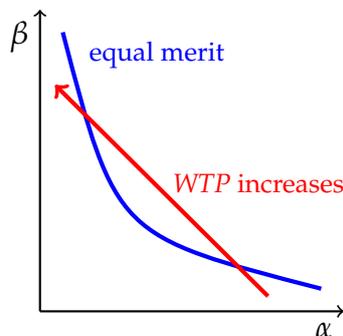
\begin{figure}[!htb]
        \centering
        \captionsetup{width=.64\linewidth}
        \begin{tikzpicture}[scale=1] 
            \draw[->,thick] (0,0) -- (4,0);
            \draw[->,thick] (0,0) -- (0,4);
            
            \node[below left] at (4,0) {$\alpha$};
            \node[below left] at (0,4) {$\beta$};
        
            \node[blue, right, font=\footnotesize] at (0.5,3.4) {equal merit};
    
            \draw[blue, ultra thick] (0.3,3.7) .. controls (1,1) .. (3.7,0.3);
            
            \draw[->, red, ultra thick] (3,0.2) -- (0.2,3);
            \node[red, right, font=\footnotesize] at (1.6,1.7) {\emph{WTP} increases};        
        \end{tikzpicture}
        \vspace{-13pt} 
        \caption{While merit increases in the north-east direction, willingness to pay increases in the north-west direction.}
        \label{fig:11}
    \end{figure}   

\subsection{Screening with ordeals}\label{screenqueueing}

I now show that the designer cannot screen equitably using only ordeals unless certain knife-edge conditions hold. Intuitively, screening with ordeals differs from screening with payments in that the allocation it produces is biased \emph{towards the poor}. That is, between two people with the same value for the good $\beta$, the poorer one will be more eager to wait to get the good, and hence will receive a higher allocation. Unlike in the case of screening with payments, this direction of bias is consistent with what the merit function requires. Notice, however, that equity constraints impose requirements not only on the \emph{direction} of the allocation's bias, but also on its exact form---the merit function $\eta(\alpha,\beta)$ specifies exactly the sets of types that have to be treated identically. However, with one screening device only, the designer will generically have “too few degrees of freedom” to pool agents in this exact way. Proposition \ref{th:2} formalizes this observation. 

\begin{proposition}\label{th:2} Suppose we do not use payments to screen, so $p\equiv 0$. Then, if a non-constant allocation rule $x(\alpha,\beta)$ is equitable and implementable, there has to exist some $\eta^* \in (\underline \eta,\overbar \eta)$ and $k \in \mathbb{R}$ such that for all types $(\alpha,\beta)$ with merit $\eta(\alpha,\beta)=\eta^*$ we have:
                \[
                \frac{\beta}{z(\alpha)} = k.
                \]
\end{proposition}

While the proof is relegated to the appendix, I provide the key intuition behind this result. Notice that the ordeal rule has to take the form $q(\alpha,\beta)\equiv \hat q(\eta(\alpha,\beta))$ since identical allocations of $x$ have to come with equal ordeals. Moreover, note that $\hat q$ has to be increasing---otherwise one could deviate and receive weakly more $x$ with strictly smaller ordeals. Now, consider the case where $\hat x$ and $\hat q$ are smoothly increasing around some merit level $\eta^*$. Then the first-order conditions for all agents with merit $\eta^*$ must hold there:
\begin{equation}\label{eq:genex}
        \text{for all }  (\alpha,\beta) \text{ such that } \eta(\alpha,\beta)=\eta^*, \quad \frac{\beta}{z(\alpha)}  = \frac{\hat q'(\eta^*)}{\hat x'(\eta^*)}.        
\end{equation}
Satisfying this condition would require equating the $MRS$s of all agents with merit $\eta^*$, of whom there are uncountably many. However, the $MRS$s of those agents are fixed, and thus, except for knife-edge cases, the designer will “lack degrees of freedom” to ensure this condition (Figure \ref{fig:12}).\footnote{As shown in Section \ref{sec:non-linear}, this fundamental problem is not due to the linearity of the utility function. Even in a richer model, the designer could only attempt equating the $MRS$s by choosing three values: the allocation and ordeal at $\eta^*$, and the value all the $MRS$s would take; this still proves insufficient.}

    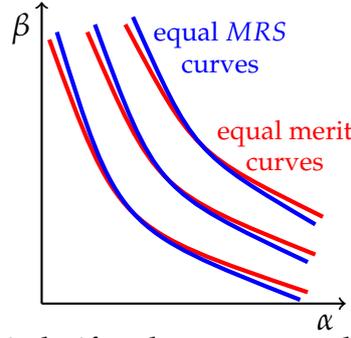
\begin{figure}[!htb]
        \centering
        \captionsetup{width=.64\linewidth}
        \begin{tikzpicture}[scale=1] 
                \draw[->,thick] (0,0) -- (4,0);
                \draw[->,thick] (0,0) -- (0,4);
                
                \node[below left] at (4,0) {$\alpha$};
                \node[below left] at (0,4) {$\beta$};
                

                \draw[red, ultra thick] (0.2,3.6) .. controls (1,1) .. (3.4,0.05);
                \draw[red, ultra thick] (0.7,3.7) .. controls (1.5,1.5) .. (3.5,0.55);
                \draw[red, ultra thick] (1.2,3.8) .. controls (2,2) .. (3.6,1.05);
            
                \draw[blue, ultra thick] (0.1,3.5) .. controls (1,1) .. (3.5,0.15);
                \draw[blue, ultra thick] (0.6,3.6) .. controls (1.5,1.5) .. (3.6,0.65);
                \draw[blue, ultra thick] (1.1,3.7) .. controls (2,2) .. (3.7,1.15);
            
                \node[red, align=center, font=\footnotesize] at (2.35,3.4) {equal $MRS$ \\ curves};
                \node[blue, align=center, font=\footnotesize] at (3.2,2.1) {equal merit \\ curves};
            \end{tikzpicture}
        \vspace{-13pt} 
        \caption{Except in knife-edge cases, no choice of $\hat x(\eta^*)$ and $\hat q(\eta^*)$ can equate the $MRS$s of all agents with merit $\eta^*$.}
        \label{fig:12}
    \end{figure}

 \subsection{Screening with payments and ordeals}

As it turns out, using both screening devices allows for rich screening without violating the equity constraint. Intuitively, every amount of the good can now be offered with a menu of payment options composed of different fee and ordeal pairs. It is crucial, however, that the designer has one screening device preferred by the poor and another preferred by the rich---two arbitrary screening devices would not suffice. When this condition is met, the designer can fine-tune these payment menus to achieve the precise bias in allocation that equity requires. Loosely speaking, being able to compose such menus solves the problem of “too few degrees of freedom” that we encountered when only ordeals were used. Thus, despite the equity constraint being extremely stringent, my main result is a positive one. 

For this result, I also impose the following technical assumption.
\begin{assumption}\label{ass:tech} 
$\lim_{\alpha\to \infty} \frac{z(\alpha)}{w(\alpha)} =0$,  $\lim_{\alpha\to -\infty} \frac{z(\alpha)}{w(\alpha)} =\infty$.
\end{assumption}
Since it imposes restrictions only on the behavior of $w,z$ outside of the type space, it has little economic content and is made for mathematical convenience.

\begin{theorem}\label{th:3} 
An allocation rule $x(\alpha,\beta)$ is equitable and implementable if and only if $x(\alpha,\beta) \equiv \hat x(\eta(\alpha,\beta))$, where $\hat x$ is increasing.
\end{theorem}

The proof proceeds in three steps. First, I reparametrize the type space so that the new types $(\kappa,\lambda)$ correspond to agents' relative values for the good and money ($\kappa$), and those for avoiding ordeals and money $(\lambda)$. I enlarge the reparametrized type space to ensure convexity.

I then consider \emph{threshold rules}, i.e. allocation rules that give $x=1$ to agents with merit above a certain level and $x=0$ to those with merit below it. I show that any threshold rule is implementable. The proof uses the characterization by \cite{ROCHET}, which states that a (multidimensional) allocation rule is implementable in a linear multidimensional setting if and only if it is a subgradient of some convex indirect utility function. I prove that such indirect utility functions can be constructed for every threshold rule by making the slope of the ordeal rule increase sufficiently quickly in $\lambda$.

Lastly, I use the fact that any allocation rule of the form $x(\alpha,\beta)=\hat x(\eta(\alpha,\beta))$ that is increasing in merit can be written as a combination of threshold rules under some probability measure. Since \eqref{eq:IC} and \eqref{eq:IR} are linear, the fact that threshold rules are implementable means that any such $x(\alpha,\beta)$ is implementable too.

\begin{proof}
I first reparametrize types:
        \[
        \kappa = \frac{\beta}{w(\alpha)}, \quad \lambda = -\frac{z(\alpha)}{w(\alpha)}, \quad \tilde \Theta =
        \left\{\left(\frac{\beta}{w(\alpha)},-\frac{z(\alpha)}{w(\alpha)}\right): \ (\alpha,\beta)\in\Theta
        \right\}.
        \]
Notice that the mapping between $(\alpha,\beta)$ and $(\kappa,\lambda)$ is one-to-one (this would not necessarily be the case if both screening devices were less costly to the rich or to the poor). Moreover, the reparametrized type space $\tilde \Theta$ is bounded, so there exist $\underline \kappa,\overbar \kappa \in \mathbb{R}$, $ \underline \lambda,\overbar \lambda\in \mathbb{R}_{--}$ such that 
$\tilde \Theta \subset [\underline \kappa,\overbar \kappa]\times [\underline \lambda,\overbar \lambda].$
Agents' utilities (up to scaling) in the reparametrized model are given by:
        \[
        \tilde U[\kappa,\lambda;x,p,q]= \kappa x+\lambda q - p.
        \]
I also rewrite the merit function in terms of $(\kappa,\lambda)$; let $\tilde \eta: \mathbb{R}\times \mathbb{R}_{--} \to \mathbb{R}$ be defined by:
\[
        \tilde \eta\left(\frac{\beta}{w(\alpha)},-\frac{z(\alpha)}{w(\alpha)}  \right) \equiv\eta(\alpha,\beta).     
\]
Assumption \ref{ass:tech} guarantees that  $\tilde \eta(\kappa,\lambda)$ is defined everywhere on $\mathbb{R}\times \mathbb{R}_{--}$.

We are now ready to prove Theorem \ref{th:3}.
By Lemma \ref{lem:inc}, any equitable and implementable allocation rule has to take the form $x(\alpha,\beta)\equiv \hat x(\eta(\alpha,\beta))$, where $\hat x$ is increasing. It therefore suffices to show that any such allocation rule is implementable on $\Theta$. In the reparametrized type space this amounts to showing that we can implement any $\tilde x(\kappa,\lambda):\tilde \Theta \to [0,1]$ such that $\tilde x(\kappa,\lambda) \equiv \hat x(\tilde \eta (\kappa,\lambda))$, where $\hat x$ is increasing. In fact, I prove a stronger statement: consider an extension of the reparametrized type space to $[\underline{\kappa}, \overline{\kappa}] \times [\underline{\lambda}, \overline{\lambda}]$. 
Let $\underline{\tilde\eta}$ and $\overbar{\tilde\eta}$ be the minimum and maximum of $\tilde \eta (\kappa,\lambda)$ on $[\underline \kappa,\overbar \kappa]\times [\underline \lambda,\overbar \lambda]$ and notice that $(\underline \eta,\overbar \eta)\subseteq [\underline{\tilde\eta}, \overbar{\tilde\eta}]$. I then show that for any increasing $\hat x: [\underline{\tilde\eta}, \overbar{\tilde\eta}]\to [0,1]$ there exists an allocation rule $\tilde x: [\underline{\kappa}, \overline{\kappa}] \times [\underline{\lambda}, \overline{\lambda}]\to [0,1]$ that is implementable on the extended type space and satisfies $\tilde x(\kappa,\lambda)=\hat x(\tilde\eta(\kappa,\lambda))$.

To that end, define a \emph{threshold rule} as a function $\hat x: [\underline{\tilde\eta}, \overbar{\tilde\eta}]\to [0,1]$ satisfying:
\begin{equation}\label{eq:examplealloc}
        \hat x(\eta)=\begin{cases}
                1 & \text{if $\eta> \eta^*$}\\
                1 \text{ or } 0 & \text{if $\eta=\eta^*$}\\
                0 & \text{otherwise},
                         \end{cases}
\end{equation}
for some $\eta^* \in [\underline{\tilde\eta}, \overbar{\tilde\eta}]$ (Figure \ref{fig:extrapol}). Let $\mathcal{T}$ be the set of threshold rules.

\begin{figure}[!htb]
        \centering
        \captionsetup{width=.64\linewidth}
        \begin{tikzpicture}[scale=1] 
                \draw[->,thick] (0,0) -- (5.5,0);
                \draw[->,thick] (0,0) -- (0,4.2);
                
                \node[below left] at (5.5,0) {$\lambda$};
                \node[below left] at (0,4.2) {$\kappa$};
            
                \node[blue, right] at (1,3.7) {merit $\eta^*$};
            
                \fill[gray!50,opacity=0.5] (0.8,0.8) rectangle (4.5,3.2);

                \begin{scope}
                        \clip (0.8,0.8) rectangle (4.5,3.2);
                        \fill[blue!30, opacity=0.5] 
                        (1.2,3.2) .. controls (1.3,3) and (1.4,3) .. (1.6,1.4)
                        .. controls (1.7,0.5) and (1.8,0.3) .. (2,0.3)
                        .. controls (2.4,0.4) and (2.6,1) .. (2.7,1.3)
                        .. controls (2.8,1.7) and (3,1.7) .. (3.2,1.6)
                        .. controls (3.5,1.4) and (3.4,1.2) .. (3.7,0.8) -- (4.5,0.8) -- (4.5,3.2) -- (1.2,3.2) -- cycle;
                    \end{scope}

                \draw[blue, ultra thick] 
                (1.2,3.2) .. controls (1.3,3) and (1.4,3) .. (1.6,1.4)
                .. controls (1.7,0.5) and (1.8,0.3) .. (2,0.3)
                .. controls (2.4,0.4) and (2.6,1) .. (2.7,1.3)
                .. controls (2.8,1.7) and (3,1.7) .. (3.2,1.6)
                .. controls (3.5,1.4) and (3.4,1.2) .. (3.7,0.8); 
            
                \draw[blue, ultra thick] (3.7,0.8) -- (4,0.5);
            
                \draw[blue, ultra thick] (0.7,3.7) -- (1.2,3.2);

                \node[below] at (4.5,0) {$\overbar \lambda$};
                \draw[black, thick, dashed] (4.5,0) -- (4.5,3.2);
            
                \node[below] at (0.8,0) {$\underline \lambda$};
                \draw[black, thick, dashed] (0.8,0) -- (0.8,3.2);
            
                \node[left] at (0,3.2) {$\overbar \kappa$};
                \draw[black, thick, dashed] (0,3.2) -- (4.5,3.2);
            
                \node[left] at (0,0.8) {$\underline \kappa$};
                \draw[black, thick, dashed] (0,0.81) -- (4.5,0.81);
            \end{tikzpicture}
        \vspace{-10pt} 
        \caption{A threshold rule allocates $x=1$ to agents with merit above $\eta^*$ (blue region) and $x=0$ to those with merit below $\eta^*$ (grey region).}
        \label{fig:extrapol}
\end{figure}
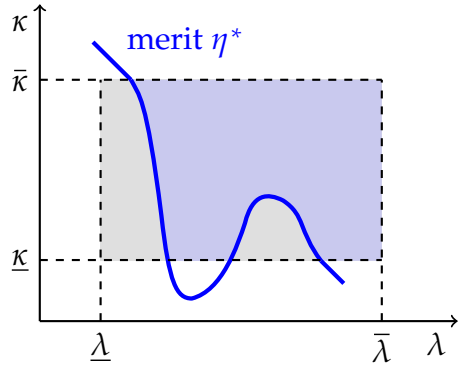

Fix any threshold rule $\hat x\in \mathcal{T}$. I will show there exists an allocation rule $\tilde x$ that is implementable on $[\underline{\kappa}, \overline{\kappa}] \times [\underline{\lambda}, \overline{\lambda}]$ and satisfies $\tilde x(\kappa,\lambda) \equiv \hat x(\tilde\eta (\kappa,\lambda))$. Following \cite{ROCHET}, the allocation rule $\tilde x(\kappa,\lambda)$ is implementable on $[\underline{\kappa}, \overline{\kappa}] \times [\underline{\lambda}, \overline{\lambda}]$ if there exists a bounded, non-negative ordeal rule $\tilde q(\kappa,\lambda)$ and a bounded, convex indirect utility function $V(\kappa,\lambda)$ such that for every  $(\kappa,\lambda)$, $[\hat x(\tilde \eta(\kappa,\lambda)) , \tilde q(\kappa,\lambda)]$ belongs to the subdifferential of $V(\kappa,\lambda)$ at that point. The following lemma (shown in the appendix) asserts that such an ordeal rule and indirect utility function indeed exist.
\begin{lemma}\label{lem:new}
There exists a bounded, non-negative ordeal rule $\tilde q(\kappa,\lambda)$ and a bounded, convex indirect utility function $V(\kappa,\lambda)$ such that $[\hat x(\tilde \eta(\kappa,\lambda)) , \tilde q(\kappa,\lambda)]$ belongs to the subdifferential of $V(\kappa,\lambda)$ for every $(\kappa,\lambda)\in [\underline \kappa,\overbar \kappa]\times [\underline \lambda,\overbar \lambda]$.
\end{lemma}  

Now, let $\mathcal{A}$ be the set of weakly increasing functions $\hat x: [\underline{\tilde\eta}, \overbar{\tilde\eta}]\to [0,1]$. $\mathcal{A}$ is convex and compact in the product topology and $\mathcal{T}$ is the set of its extreme points. Hence, by Choquet's theorem, for every $\hat x^*\in \mathcal{A}$ there exists a probability measure $\mu$ on $\mathcal{T}$ such that $\hat x^* = \int_{\mathcal{T}} \hat x \mu(d\hat x)$ \citep{phelps2001lectures}. However, we have already shown that for every $\hat x\in \mathcal{T}$ there exists an allocation rule $\tilde x[\hat x]$ satisfying $\tilde x[\hat x](\kappa,\lambda)\equiv \hat x(\tilde \eta(\kappa,\lambda))$ that is implementable on $[\underline{\kappa}, \overline{\kappa}] \times [\underline{\lambda}, \overline{\lambda}]$. Define:
\[
\tilde x^* := \int_{\mathcal{T}} \tilde x[\hat x] \mu(d \hat x) = \int_{\mathcal{T}} \hat x \mu(d \hat x)= \hat x^*.
\]
Since all $\tilde x[\hat x]$ are implementable and \eqref{eq:IC} and \eqref{eq:IR} are linear in $x,p$ and $q$, it follows that $\tilde x^*$ is implementable too.  
\end{proof}

\section{Observable wealth}\label{sec:obswealth}

While I assumed that neither need nor wealth are observable, the designer usually has some information about them. For instance, tax data proxies for one's wealth, even if some income sources or assets remain unobserved. In such cases, agents' private information can be thought of as \emph{residual uncertainty} after accounting for observables. Such uncertainty can still be substantial even in means-tested programs (see footnote \ref{footnote:2}).

Still, what can the designer do if she perfectly observes agents' need for money $\alpha$? As it turns out, the set of equitable allocation rules that are implementable with one instrument when $\alpha$ is observed is identical to the set of equitable allocation rules that are implementable with two instruments when $\alpha$ is private. Intuitively, when the designer observes $\alpha$, she can “control for” the fact that some agents prefer the good because they are wealthy and screen purely based on need.
\begin{proposition}
Suppose need for money $\alpha$ is observable and the designer uses either only payments or only ordeals to screen. Then the allocation rule $x(\alpha,\beta)$ is equitable and implementable if and only if $x(\alpha,\beta) \equiv \hat x(\eta(\alpha,\beta))$, where $\hat x$ is increasing.
\end{proposition}
\begin{proof}
Since $\alpha$ is observable, the allocation rule can be implemented separately for every value of need for money $\alpha$. Thus, the allocation rule $x(\alpha,\beta)$ is implementable if and only if it is weakly increasing in the value for the good $\beta$. Recall the equity constraint is satisfied if and only if $x(\alpha,\beta)\equiv \hat x(\eta(\alpha,\beta))$. Since $\eta(\alpha,\beta)$ is strictly increasing in $\beta$, $x(\alpha,\beta) \equiv \hat x(\eta(\alpha,\beta))$ is equitable and implementable if and only if $\hat x$ is weakly increasing.
\end{proof}

\section{Measures of equity violation}\label{sec:relax}

My model of societal perceptions of equity was built around a single merit function. In reality, however, people often share general principles concerning equity and desert but hold different opinions about finer trade-offs or ways in which these principles should be applied. Real-world equity constraints on policymakers would therefore be less demanding than my analysis suggests. While Theorem \ref{th:3} says that rich screening with multiple instruments is possible even under such overly strong restrictions, it is also interesting to compare mechanisms screening with only payments and only ordeals when the equity constraint is relaxed.

A na\"ive intuition suggests that screening with ordeals violates equity less as it biases the allocation “in the right direction”: among agents with the same need, it allocates more to those with \emph{lower} wealth, and thus to those with higher merit. However, as discussed in Subsection \ref{screenqueueing}, a mechanism screening with ordeals will generically fail to pool together agents in a way that “matches the shape” of the merit curve. Thus, in some cases this “shape effect” could dominate the aforementioned “direction effect”. Still, one could expect the direction effect to dominate when merit is primarily determined by one's wealth, that is, when $\eta_\alpha \gg \eta_\beta$. This is plausible in some settings: while we might expect deservingness of public housing to depend strongly on both one's income and housing situation, the deservingness of a free school lunch is likely to depend almost entirely on wealth. As it turns out, however, even this intuition is not fully accurate. This section shows that comparing the sizes of equity violations is more subtle, and depends crucially on the way equity violations are evaluated. Indeed, I show that even simple measures of equity violation based on the principle that equally deserving individuals should be treated equally (see, e.g., \cite{dwork2012fairness}) may deliver the opposite conclusion. This perhaps surprising result illustrates that the equity ranking of screening mechanisms is highly sensitive to how violations of equal treatment are quantified, and that intuitive “directional” arguments can be overturned under fairness metrics built on appealing principles. It therefore points to a broader conceptual challenge of specifying normatively appealing fairness metrics.

I now define the aforementioned notion. To that end, I first develop a measure of \emph{equity violation} capturing how far away a particular allocation is from satisfying the equity constraint (here, I interpret the merit function as a “rough consensus” among the public). My measure assumes that every agent assesses the allocation's equity by looking at agents similar to herself and comparing those with the same allocation as her to those with the same merit as her. In other words, she compares agents similar to herself who \emph{are} treated the same as her with those who \emph{should} be treated the same. The further apart these two sets are, the more inequitable the allocation seems to her. Then, the degree to which the whole allocation violates equity is the size of the largest such “local equity violation”. To formalize this notion, let us define $\angle: \mathbb{R}\cup \{-\infty,\infty\} \to [0,\pi]$ as the angle between the horizontal axis and a line with a given slope.\footnote{Formally, we have:
\[
\angle(m) =
\begin{cases}
\arctan(m), & m \ge 0,\\[4pt]
\arctan(m) + \pi, & m < 0.
\end{cases}
\]
While my results do not require this exact functional form, it lets us visualize the size of the local equity violation.
} 

\begin{definition}[Equity violation]\label{def:eqviol}
For every type $(\alpha,\beta)\in \Theta$, let $D(\alpha,\beta)$ be the set of directions in which the allocation is locally constant:
\[
D(\alpha,\beta) = \{d\in \mathbb{R}^2: \ \nabla_d x(\alpha,\beta) =0 \}.
\]
Let the \textbf{local equity violation} for type $(\alpha,\beta)$ be:
\begin{equation}
        l(\alpha,\beta) = \begin{cases}
            \infty & \text{if }D(\alpha,\beta) = \emptyset,\\
            \inf_{(d_1,d_2)\in D(\alpha,\beta)} \Big|\angle\left(\frac{d_2}{d_1}\right) - \angle\left(-\frac{\eta_\alpha(\alpha,\beta)}{\eta_\beta(\alpha,\beta)}\right) \Big|
            & \text{otherwise}.\footnotemark
        \end{cases}
    \end{equation}
The \textbf{equity violation} of the allocation rule $x(\alpha,\beta)$, denoted $L(x)$, is its largest local equity violation:
\[
L(x):= \sup_{(\alpha,\beta) \in \Theta} l(\alpha,\beta).
\]
\end{definition}
To build intuition, consider a smooth allocation rule that is strictly increasing in need for the good $\beta$ and fix some type $(\alpha^a,\beta^a)$. Then all types with allocations equal to that of $(\alpha^a,\beta^a)$ will lie on a smooth curve passing through $(\alpha^a,\beta^a)$. Figure \ref{fig:2subfig1} illustrates such a curve together with this type's \emph{iso-merit curve}, that is, the set of types with the same merit as $(\alpha^a,\beta^a)$. Since both of these curves are smooth, we can compare “how far apart” they are in the neighborhood of $(\alpha^a,\beta^a)$ by looking at the angle between their tangents there---this angle measures the \emph{local equity violation} at $(\alpha^a,\beta^a)$ (Figure \ref{fig:2subfig2}).

\begin{figure}[h!]
        \centering
        \captionsetup{width=.64\linewidth}
        \begin{subfigure}{0.33\linewidth}
                \begin{tikzpicture}[scale=1] 
                        \draw[->,thick] (0,0) -- (4,0);
                        \draw[->,thick] (0,0) -- (0,4);
                        
                        \node[below left] at (4,0) {$\alpha$};
                        \node[below left] at (0,4) {$\beta$};
                        
        
                        \draw[blue, ultra thick] (0.7,3.7) .. controls (1.5,1.5) .. (3.5,0.55);
                        \draw[red, ultra thick] (0.6,3.2) .. controls (1.3,1.5) .. (3.6,0.9);
        
                        \fill[blue] (1.9,1.4) circle (3pt);
                        \node[blue, below left, font=\small] at  (2,1.4) {$(\alpha^a,\beta^a)$};
                        \node[red,  font=\small] at (3.2,2) {allocation};
                        \node[red,  font=\small] at (3.2,1.6) {$x(\alpha^a,\beta^a)$};

                        \node[blue, right, font=\small] at (0.9,3.4) {merit $\eta(\alpha^a,\beta^a)$};
                      \end{tikzpicture}
            \caption{Figure \ref{fig:2subfig1}: The iso-merit curve and the iso-allocation curve at type $(\alpha^a,\beta^a)$.}
            \label{fig:2subfig1}
        \end{subfigure}
        \ \ \
        \begin{subfigure}{0.33\linewidth}
                \begin{tikzpicture}[scale=1]
                        \draw[->,thick] (0,0) -- (4,0);
                        \draw[->,thick] (0,0) -- (0,4);
                        
                        \node[below left] at (4,0) {$\alpha$};
                        \node[below left] at (0,4) {$\beta$};
                        
                        \draw[green!50!black, line width=1.5mm] (1.9,2.1) arc (155:133:0.4);
                        \draw[blue, ultra thick] (0.7,3.7) -- (3.5,0.55);
                        \draw[red, ultra thick] (0.6,2.7) -- (3.6,1.28);
                        
                        \fill[blue] (2.3,1.9) circle (3pt);
                        \node[blue, above right, font=\small] at  (2.3,1.7) {$(\alpha^a,\beta^a)$};
                        \node[red,  font=\small] at (1.2,1.9) {allocation};
                        \node[red,  font=\small] at (1.2,1.5) {$x(\alpha^a,\beta^a)$};
                        \node[blue, right, font=\small] at (1.2,3.4) {merit $\eta(\alpha^a,\beta^a)$};

                        \node[green!50!black, above right, font=\small] at (1.7,2.3) {$l(\alpha^a,\beta^a)$};
                    \end{tikzpicture}
            \caption{Figure \ref{fig:2subfig2}: The slopes of the iso-merit curve and the iso-allocation curve at type $(\alpha^a,\beta^a)$.}
            \label{fig:2subfig2}
        \end{subfigure}
        \label{fig:mainfig}
\end{figure}

Note that the slope of the iso-merit curve, given by $-\eta_\alpha/\eta_\beta$, is going to be \emph{more steep} when merit depends more strongly on wealth than it does on need, i.e. when $\eta_\alpha$ is large relative to $\eta_\beta$. Conversely, the iso-merit curve will be flat when merit is mostly determined by need. The following proposition then says that, as long as iso-merit curves are sufficiently flat everywhere, there is a non-constant allocation rule implemented with only ordeals that violates equity by less than any non-constant allocation rule implemented using only payments. In that sense, ordeals violate equity by less in cases where merit is \emph{weakly} dependent on wealth. Moreover, as discussed later, payments sometimes violate equity by less than ordeals even when merit is strongly tied to wealth. Perhaps surprisingly, this contradicts the aforementioned ``directional" intuition.

\begin{proposition}\label{prop:relax}
Assume ordeals are bounded above by some $\overbar q$. Then there exists $M<0$ with the following property: suppose that the slope of the iso-merit curve is flatter than $M$ at every type, so:
        \[
\text{for all }(\alpha,\beta)\in \Theta, \quad        -\frac{\eta_\alpha(\alpha,\beta)}{\eta_\beta(\alpha,\beta)} > M.
        \]
Then there exists a non-constant allocation rule $x_q$ that is implementable with only ordeals that produces a strictly smaller equity violation than any non-constant allocation rule $x_p$ that is implementable with only payments: $L(x_q)< L(x_p)$.
        \end{proposition}

While the proof is relegated to the appendix, I illustrate its key intuition with the case of a smooth allocation that increases in the value for the good $\beta$. Fix some type $(\alpha^a,\beta^a)$ and compare two allocations: $x_p$ implemented using only payments and $x_q$ implemented using only ordeals. Like before, the sets of agents with the same allocation as $(\alpha^a,\beta^a)$ will be smooth curves passing through that point. Moreover, the insights from Subsections \ref{screeningpayments} and \ref{screenqueueing} tell us that the curve for the allocation rule $x_p$ will be upwards-sloping (Figure \ref{fig:3subfig1}), while the curve for the allocation rule $x_q$ will be downwards-sloping (Figure \ref{fig:3subfig2}).

\begin{figure}[h!]
        \centering
        \begin{subfigure}{0.33\linewidth}
                \begin{tikzpicture}[scale=1] 
                        \draw[->,thick] (0,0) -- (4,0);
                        \draw[->,thick] (0,0) -- (0,4);
                        
                        \node[below left] at (4,0) {$\alpha$};
                        \node[below left] at (0,4) {$\beta$};
                        
                        \draw[blue, ultra thick] (0.7,3.7) .. controls (1.5,1.5) .. (3.5,0.55);
                        \draw[red, ultra thick] (3.6,3.6) .. controls (2.4,1.6) .. (0.6,0.65);
                        
                        \fill[blue] (1.9,1.4) circle (3pt);
                        \node[blue, right, font=\small] at  (2.1,1.4) {$(\alpha^a,\beta^a)$};
                        
                        \node[blue, font=\small] at (0.85,1.83) {merit};
                        \node[blue, font=\small] at (0.85,1.43) {$\eta(\alpha^a,\beta^a)$};
                        \node[red,  font=\small] at (2.32,3.5) {allocation};
                        \node[red,  font=\small] at (2.32,3.1) {$x(\alpha^a,\beta^a)$};
                        
                    \end{tikzpicture}
            \caption{Figure \ref{fig:3subfig1}: Allocation rule implemented using only payments}
            \label{fig:3subfig1}
        \end{subfigure}
        \ \ \ 
        \begin{subfigure}{0.33\linewidth}
                \begin{tikzpicture}[scale=1] 
                        \draw[->,thick] (0,0) -- (4,0);
                        \draw[->,thick] (0,0) -- (0,4);
                        
                        \node[below left] at (4,0) {$\alpha$};
                        \node[below left] at (0,4) {$\beta$};
                        
        
                        \draw[blue, ultra thick] (0.7,3.7) .. controls (1.5,1.5) .. (3.5,0.55);
                        \draw[red, ultra thick] (0.6,3.2) .. controls (1.3,1.5) .. (3.6,0.9);
        
                        \fill[blue] (1.9,1.4) circle (3pt);
                        \node[blue, below left, font=\small] at  (2,1.4) {$(\alpha^a,\beta^a)$};
                        \node[red,  font=\small] at (3.2,2) {allocation};
                        \node[red,  font=\small] at (3.2,1.6) {$x(\alpha^a,\beta^a)$};

                        \node[blue, right, font=\small] at (0.9,3.4) {merit $\eta(\alpha^a,\beta^a)$};
                      \end{tikzpicture}
            \caption{Figure \ref{fig:3subfig2}: Allocation rule implemented using only ordeals}
            \label{fig:3subfig2}
        \end{subfigure}
        \label{fig:mainfig}
    \end{figure}

Let us now compare the local equity violations of these two allocation rules. Figure  \ref{fig:4subfig1} illustrates that if the iso-merit curve is sufficiently flat, its angle with the iso-allocation curve for $x_q$ will be smaller than that with the iso-allocation curve for $x_p$. If we can impose a sufficiently low uniform bound on the slopes of iso-merit curves, this will be true for every type and every allocation rule $x_p$ implemented with only payments.  Figure \ref{fig:4subfig2}, on the other hand, illustrates why the result of Proposition \ref{prop:relax} fails when iso-merit curves are not flat enough.

\begin{figure}[h!]
        \centering
        \captionsetup{width=.64\linewidth}
        \begin{subfigure}{0.33\linewidth}
                \begin{tikzpicture}[scale=1]
                        \draw[->,thick] (0,0) -- (4,0);
                        \draw[->,thick] (0,0) -- (0,4);
                        
                        \draw[dashed] (2.3,0) -- (2.3,3.9);
                        \draw[dashed] (0,1.9) -- (4,1.9);

                        \node[below left] at (4,0) {$\alpha$};
                        \node[below left] at (0,4) {$\beta$};
                        
                        \draw[green!50!black, line width=1.5mm] (1.05,2.2) arc (165:157:2);
                        \draw[green!50!black, line width=1.5mm] (1.6,2.07) arc (161:82:0.83);
                        \draw[red, ultra thick] (0.6,2.7) -- (3.6,1.28);
                        \draw[blue, ultra thick] (0.4,2.4) -- (3.7,1.5);
                        \draw[brown, ultra thick] (1.9,0.4) -- (2.75,3.6);
                    
                        \fill[blue] (2.3,1.9) circle (3pt);
                        \node[red,  font=\small] at (1.1,3) {$x_q(\alpha^a,\beta^a)$};
                        \node[brown,  font=\small] at (3.52,2.6) {$x_p(\alpha^a,\beta^a)$};

                    \end{tikzpicture}
            \caption{Figure \ref{fig:4subfig1}}
            \label{fig:4subfig1}
        \end{subfigure}
        \ \ \ 
        \begin{subfigure}{0.33\linewidth}
                \begin{tikzpicture}[scale=1]
                        \draw[->,thick] (0,0) -- (4,0);
                        \draw[->,thick] (0,0) -- (0,4);
                        \draw[dashed] (2.3,0) -- (2.3,3.9);
                        \draw[dashed] (0,1.9) -- (4,1.9);
                        \node[below left] at (4,0) {$\alpha$};
                        \node[below left] at (0,4) {$\beta$};
                        
                        \draw[green!50!black, line width=1.5mm] (1.05,2.2) arc (165:107:1.5);
                        \draw[green!50!black, line width=1.5mm] (2.2,2.5) arc (100:80:0.8);
                        \draw[blue, ultra thick] (1.95,3.7) -- (2.55,0.55);
                        \draw[red, ultra thick] (0.4,2.4) -- (3.7,1.5); 
                        \draw[brown, ultra thick] (1.9,0.4) -- (2.75,3.6);
                       
                        \fill[blue] (2.3,1.9) circle (3pt);
                        \node[red,  font=\small] at (1.2,1.6) {$x_q(\alpha^a,\beta^a)$};
                        \node[brown,  font=\small] at (3.47,2.6) {$x_p(\alpha^a,\beta^a)$};
                    \end{tikzpicture}
                    \caption{Figure \ref{fig:4subfig2}}
            \label{fig:4subfig2}
        \end{subfigure}
        \caption{If the iso-merit curve is sufficiently steep, screening with money may be more equitable than screening with ordeals.}
        \label{fig:mainfig}
    \end{figure}
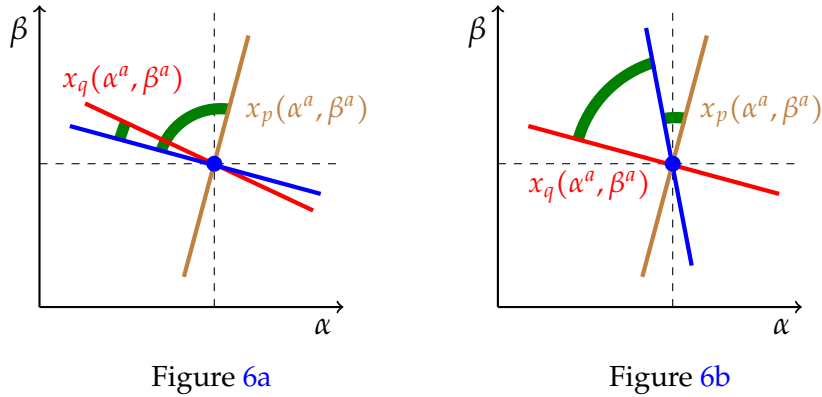

\section{Non-linear utility functions}\label{sec:non-linear}

While the main section considers a linear utility specification, many insights from the model do not rely on this structure. In this section, I therefore consider the following richer specification:
\[
U[\alpha,\beta;x,p,q] = v(\beta,x)-w(\alpha,p)-z(\alpha,q),
\]
where the following conditions hold for all $(\alpha,\beta)\in \mathbb{R}^2$ and all $x\in [0,1], p\in\mathbb{R}, q\in \mathbb{R}_+$:
\begin{enumerate}
        \item $v, w,z$ are twice continuously differentiable.
        \item $v_x>0$, \   $w_p>0$,\  $z_q>0$.
        \item $v_{\beta x}>0, \  w_{\alpha p}>0, \  \ z_{\alpha q}<0$.
        \item $v(\beta,0) = w(\alpha,0)=z(\alpha,0)=0$.
\end{enumerate}
Thus, the linearity assumption is replaced by weaker single-crossing conditions. Indeed, note that the baseline model is a special case of the one considered in this section.

In this setting, Proposition \ref{th:1} about the impossibility of equitable screening with money alone (as well as its proof) remains the same. Proposition \ref{th:2} requires the following modification:
\newtheorem*{propositiontwostar}{Proposition 2$^\ast$}
{
\begin{propositiontwostar}\label{th:2prime}
Suppose we do not use payments to screen, so $p\equiv 0$. Then, if a non-constant allocation rule $x(\alpha,\beta)$ is equitable and implementable, one of the following conditions has to hold for some $\eta^* \in (\underline \eta,\overbar \eta)$:
\begin{enumerate}
    \item There exist $x\in [0,1]$, $q \in \mathbb{R}_+$ and $k \in \mathbb{R}$ such that for all types $(\alpha,\beta)$ with merit $\eta(\alpha,\beta)=\eta^*$,
    \[
    \frac{v_x(\beta,x)}{z_q(\alpha,q)} = k.
    \]
    \item There exist $x^a \neq x^b\in [0,1]$, $q^a,q^b \in \mathbb{R}_+$ such that for all types $(\alpha,\beta)$ with merit $\eta(\alpha,\beta)=\eta^*$,
    \[
    v(\beta,x^a) - z(\alpha,q^a) = v(\beta,x^b) - z(\alpha,q^b).
    \]
\end{enumerate}
\end{propositiontwostar}
}
Proposition \ref{prop:relax} also carries through \emph{verbatim}. Unfortunately, however, Theorem \ref{th:3} does not easily extend. This is for two reasons: first, the non-linear utility specification lacks a workable characterization of implementability similar to that of \cite{ROCHET}. Second, non-linearities preclude the application of a mixture argument used in the proof of the result. Nevertheless, I conjecture that a result similar to Theorem \ref{th:3} is true also in this richer setting.

\section{Discussion}\label{sec:conc}

While my approach to modeling perceived equity is highly stylized, it offers general qualitative conclusions. First, every screening instrument will bias the allocation towards the group for which this instrument is less costly---this makes screening with payments problematic from an equity standpoint. Using a different instrument (like waiting in line) could reverse this bias but the designer's control over the allocation would still be limited. Consequently, the resulting bias might still not satisfy the public's equity concerns. I show that this problem can be solved by combining multiple screening instruments which on their own favor different social groups. Doing so gives the designer freedom to tinker with various groups' differential cost of the allocated good, and therefore to improve efficiency through screening while still producing an allocation that is seen as fair. In practice, this could be achieved by combining payment-based mechanisms, such as congestion or emission charges, with possibly laborious rebate processes---ordeals poorer and more car-reliant individuals could go through to avoid the disproportionate burden of payment. Finally, I show that even once the merit function is fixed, evaluating the “degree” of equity violation is itself delicate: plausible ways of aggregating deviations from equal treatment can overturn appealing intuitions about which policies are fairer.

\appendix
\setstretch{1} 

\section{Omitted proofs}

When providing the proofs of Proposition \ref{th:1} and \ref{prop:relax} as well as related lemmas, I show them for the extended model introduced in Section \ref{sec:non-linear}. I also write $(\alpha^a_\delta,\beta^a_\delta)$ for $(\alpha^a+\delta,\beta^a+\delta)$,  $x^a_\delta$ for $x(\alpha^a_\delta,\beta^a_\delta)$ and use $p^a_\delta$, $q^a_\delta$ and $\eta^a_\delta$ analogously. I omit the subscript when $\delta=0$.

The following lemma will be useful in the proofs presented here. While it resembles a first-order condition, its proof is more involved because the type space is two-dimensional and the ordeal rule need not be piecewise differentiable.
\begin{lemma}[Generalized FOC]\label{lem:MRS}
Suppose we do not use payments to screen, so $p\equiv 0$. Fix any type $(\alpha^a,\beta^a)\in \Theta$ and assume there exists a decreasing sequence $\{\delta_i\}_i$ such that $\delta_i \to 0$ for which $\{x^a_{\delta_i}\}_i$ is strictly decreasing and $x^a_{\delta_i} \to x^a$. Then, for any $(\alpha^b,\beta^b)$ such that $x^b= x^a$ we have:
\begin{equation}\label{lem:eqMRS}
        \frac{v_x(\beta^b,x^a)}{z_q(\alpha^b,q^a)} \leq \frac{v_x(\beta^a,x^a)}{z_q(\alpha^a,q^a)}.
\end{equation}
\end{lemma}
\begin{proof}
Note $\{q^a_{\delta_i}\}_i$ is strictly decreasing since strictly lower allocations must come with strictly lower ordeals. I first show that $q^a_{\delta_i} \to q^a$. Indirect utility has to be continuous at $(\alpha^a,\beta^a)$, so:
\[
v(\beta^a,  x^a) - z(\alpha^a, q^a)= \lim_{i\to \infty } \left\{v(\beta^a_{\delta_i}, x^a_{\delta_i}) - z(\alpha^a_{\delta_i}, q^a_{\delta_i}) \right\}.        
\]
Note that $\lim_{i\to \infty } v(\beta^a_{\delta_i}, x^a_{\delta_i}) = v(\beta^a,  x^a)$. Since $\alpha^a_{\delta_i}\to \alpha^a$ and $z$ is continuous and strictly increasing in the latter argument, we have $q^a_{\delta_i}\to q^a$.  

Now, suppose \eqref{lem:eqMRS} fails:
\[
        \frac{v_x(\beta^b,x^a)}{z_q(\alpha^b,q^a)} > \frac{v_x(\beta^a,x^a)}{z_q(\alpha^a,q^a)}.        
\]
By continuity of $v_x$ and $z_q$, for $i$ high enough we have:
        \begin{equation}\label{eq:bbb}
                \frac{v_x(\beta^b,x^a)}{z_q(\alpha^b,q^a)} > \frac{v_x(\beta^a_{\delta_i},x^a)}{z_q(\alpha^a_{\delta_i},q^a)}.        
        \end{equation}
Since $\{x^a_{\delta_i}\}_i$ and $\{q^a_{\delta_i}\}_i$ are strictly decreasing and tend to $x^a$ and $q^a$, we have that for all $\alpha,\beta$:
\begin{equation}\label{eq:convergence}
\frac{v(\beta,x^a_{\delta_j})-v(\beta,x^a)}{x^a_{\delta_j}-x^a} \to v_x(\beta,x^a), \quad \frac{z(\alpha,q^a_{\delta_j})-z(\alpha,q^a)}{q^a_{\delta_j}-q^a} \to z_q(\alpha,q^a)    \quad \text{ as $j\to\infty$.}
\end{equation}
 Fix $i$ high enough that \eqref{eq:bbb} holds. Then, by \eqref{eq:convergence}, for $j$ high enough we have:
        \[
        \frac
        {
        \frac{v(\beta^b,x^a_{\delta_j})-v(\beta^b,x^a)}{x^a_{\delta_j}-x^a}
        }
        {
        \frac{z(\alpha^b,q^a_{\delta_j})-z(\alpha^b,q^a)}{q^a_{\delta_j}-q^a}
        }
        >
        \frac
        {
        \frac{v(\beta^a_{\delta_i},x^a_{\delta_j})-v(\beta^a_{\delta_i},x^a)}{x^a_{\delta_j}-x^a}
        }
        {
        \frac{z(\alpha^a_{\delta_i},q^a_{\delta_j})-z(\alpha^a_{\delta_i},q^a)}{q^a_{\delta_j}-q^a}.
        }
        \] 
In particular, take $j>i$ and notice that increasing $i$ to $j$ further relaxes the inequality. Thus:
        \begin{equation}\label{eq:ccc}
                \frac
                {
                \frac{v(\beta^b,x^a_{\delta_j})-v(\beta^b,x^a)}{x^a_{\delta_j}-x^a}
                }
                {
                \frac{z(\alpha^b,q^a_{\delta_j})-z(\alpha^b,q^a)}{q^a_{\delta_j}-q^a}
                }
                >
                \frac
                {
                \frac{v(\beta^a_{\delta_j},x^a_{\delta_j})-v(\beta^a_{\delta_j},x^a)}{x^a_{\delta_j}-x^a}
                }
                {
                \frac{z(\alpha^a_{\delta_j},q^a_{\delta_j})-z(\alpha^a_{\delta_j},q^a)}{q^a_{\delta_j}-q^a}.
                }     
        \end{equation}
Now, by revealed preference we have:
        \[
        v(\beta^a_{\delta_j},x^a_{\delta_j}) - z(\alpha^a_{\delta_j},q^a_{\delta_j}) \geq v(\beta^a_{\delta_j},x^a)  - z(\alpha^a_{\delta_j},q^a)  
        \quad
        \Longrightarrow 
        \quad
        \frac
        {
        \frac{v(\beta^a_{\delta_j},x^a_{\delta_j})-v(\beta^a_{\delta_j},x^a)}{x^a_{\delta_j}-x^a}
        }
        {
        \frac{z(\alpha^a_{\delta_j},q^a_{\delta_j})-z(\alpha^a_{\delta_j},q^a)}{q^a_{\delta_j}-q^a}
        }
        \geq 
        \frac{q^a_{\delta_j}-q^a}{x^a_{\delta_j}-x^a}.
        \]
Combining the latter inequality with \eqref{eq:ccc} gives:
\[
        \frac
        {
        \frac{v(\beta^b,x^a_{\delta_j})-v(\beta^b,x^a)}{x^a_{\delta_j}-x^a}
        }
        {
        \frac{z(\alpha^b,q^a_{\delta_j})-z(\alpha^b,q^a)}{q^a_{\delta_j}-q^a}
        }
        >
        \frac{q^a_{\delta_j}-q^a}{x^a_{\delta_j}-x^a}
        \quad
        \Longrightarrow
        \quad
        v(\beta^b,x^a_{\delta_j}) - z(\alpha^b,q^a_{\delta_j}) > v(\beta^b,x^a) - z(\alpha^b,q^a).
\]
Hence, $x^b\neq x^a$; contradiction.
\end{proof}

\subsection{Proof of Proposition \ref{th:2}}

Let $x(\alpha,\beta)\equiv \hat x(\eta(\alpha,\beta))$ be an equitable, implementable and non-constant allocation rule and consider two cases.

\textbf{Case 1: }$\hat x$ is discontinuous at some $\eta^a$. Indirect utility has to be continuous so, by continuity of $\eta$, $v$ and $z$, the following holds for any type $(\alpha,\beta)$ with merit $\eta^a$: 
\[
        v(\beta,  \hat x_+(\eta^a)) - z(\alpha,  \hat q_+(\eta^a)) =  v(\beta,  \hat x_-(\eta^a)) - z(\alpha,  \hat q_-(\eta^a)).            
\]
However, the existence of $\hat x_+(\eta^a),\hat x_-(\eta^a)$, $\hat q_+(\eta^a),\hat q_-(\eta^a)$ implies condition $2.$ in Proposition \ref{th:2}.

\textbf{Case 2:} $\hat x$ is continuous. The proof of this case relies on the Generalized FOC (Lemma \ref{lem:MRS}). I first show that there exists a sequence that lets us apply it.

\begin{fact}\label{fact:sequence}
There exists $\eta^a\in (\underline \eta,\overbar \eta)$ and a decreasing sequence $\{f_i\}_i$ of $f_i\in (\underline \eta,\overbar \eta)$ such that $f_i \to \eta^a$ and $\{\hat x(f_i)\}_i$ is strictly decreasing.
        \end{fact}
        \begin{proof}
Recall that $\hat x$ is not constant. By Lemma \ref{lem:inc} it is also weakly increasing, so there exist $\eta^b,\eta^c\in (\underline \eta,\overbar \eta)$ such that $\eta^b<\eta^c$ and $\hat x(\eta^b) <  \hat x(\eta^c)$. Let $\eta^a = \sup\{\eta: \hat x(\eta) = \hat x(\eta^b)\}$. Since $\hat x$ is continuous, $\hat x(\eta^a) = \hat x(\eta^b)$. Now, take any decreasing sequence $\{e_i \}_i$ such that for every $i$, $e_i \in (\underline \eta,\overbar \eta)$ and $e_i \to \eta^a$. Since $\hat x$ is weakly increasing, $\hat x(e_i)\leq \hat x(e_j)$ whenever $i>j$. Also, by continuity of $\hat x$, $\hat x(e_i)\to \hat x(\eta^a)$ and, by the construction of $\eta^a$, $\hat x(\eta^a)< \hat x(e_i)$ for every $i$. Since $\hat x$ is continuous, there exists a subsequence $\{f_i\}_i$ of $\{e_i \}_i$ for which $\hat x(f_i)$ is decreasing strictly.
        \end{proof}
Take $\eta^a$, $\{f_i\}_i$ from the statement of Fact \ref{fact:sequence} and consider any $(\alpha^b,\beta^b)$ with merit $\eta^a$. Since $f_i \to \eta^a$ and $\hat x$ is continuous, $\hat x(f_i)\to \hat x(\eta^a)$. Then, by continuity and monotonicity of $\eta(\alpha,\beta)$, for all $i$ high enough there exist $\delta^b_i$ such that $f_i =\eta(\alpha^b_{\delta_i},\beta^b_{\delta_i})$, $\delta^b_i \to 0$ as $f_i \to \eta^a$, and $\{\delta^b_i\}_i$ is decreasing. Since $(\alpha^b,\beta^b)$ has the same allocation $\hat x(\eta^a)$ as $(\alpha^a,\beta^a)$, Lemma \ref{lem:MRS} gives:
\[
\frac{v_x(\beta^a,x^a)}{z_q(\alpha^a,q^a)} \leq \frac{v_x(\beta^b,x^a)}{z_q(\alpha^b,q^a)} .
\]
However, notice that such a sequence $\{\delta^a_i\}_i$ can also be found for type $(\alpha^a,\beta^a)$, and thus Lemma \ref{lem:MRS} gives the reverse inequality too. Hence, all types with merit $\eta^a$ must the have same value of $v_x(\beta,x^a)/z_q(\alpha,x^a)$, which implies condition $1.$ in Proposition \ref{th:2}.

\subsection{Proof of Lemma \ref{lem:new}}

We first show $\tilde \eta(\kappa,\lambda)$ is twice continuously differentiable on $\mathbb{R}\times \mathbb{R}_{--}$, and that $\tilde \eta_\kappa(\kappa,\lambda)$ is positive and uniformly bounded away from zero on $\mathbb{R}\times [\underline \lambda,\overbar \lambda]$. Note $-\frac{z(\alpha)}{w(\alpha)}$ is strictly increasing and has a twice continuously differentiable inverse $f:\mathbb{R}_{--}\to\mathbb{R}$. Thus:
\[
\tilde \eta(\kappa, \lambda) = \eta\left(f(\lambda), \ \kappa \cdot w(f(\lambda))\right).
\]
Calculation confirms that $\tilde \eta(\kappa,\lambda)$ is twice continuously differentiable.

I now show that $\tilde \eta_\kappa(\kappa,\lambda)$ is positive and bounded away from zero on $\mathbb{R}\times [\underline \lambda,\overbar \lambda]$. Notice:
\[
\tilde \eta_\kappa(\kappa, \lambda) = \eta_\beta\left(f(\lambda), \ \kappa \cdot w(f(\lambda))\right)\cdot w(f(\lambda)).
\]
$\eta_\beta(\alpha,\beta)$ is strictly positive and uniformly bounded away from zero on $\mathbb{R}^2$ by point $2.$ of Assumption \ref{ass:tech}. $w(f(\lambda))$ is strictly positive everywhere. Moreover, $w$ and $f$ are continuous, so $w(f(\lambda))$ attains its (strictly positive) minimum on $[\underline \lambda,\overbar \lambda]$. Thus, $\tilde \eta_\kappa(\kappa,\lambda)$ is positive and uniformly bounded away from zero on $\mathbb{R}\times [\underline \lambda,\overbar \lambda]$.

Now, define $\kappa^*:\mathbb{R}_{--}\to\mathbb{R}$ such that:
\begin{equation}\label{eq:kappadef}
        \eta^* = \tilde \eta (\kappa^*(\lambda),\lambda).
\end{equation}
That is, $\kappa^*(\lambda)$ gives the value of $\kappa$ for which type $(\kappa,\lambda)$ attains the threshold merit level $\eta^*$ (see Figure \ref{fig:extrapol}). Notice $\kappa^*(\lambda)$ is well-defined everywhere on $\mathbb{R}_{--}$. This is because $\tilde \eta_\kappa(\kappa,\lambda)$ is positive and uniformly bounded away from zero so, for every $\lambda^*\in \mathbb{R}_{--}$, $\tilde \eta (\kappa,\lambda^*)$ attains every value in $\mathbb{R}$ for some $\kappa$.

Since $\tilde \eta(\kappa,\lambda)$ is twice continuously differentiable and $\tilde \eta(\kappa,\lambda)>0$, implicitly differentiating \eqref{eq:kappadef} tells us that $\kappa^*(\lambda)$ is twice continuously differentiable too. Thus, the following bounds are finite:
\[
        M_1:=  \max_{\lambda \in [\underline{\lambda}, \overline{\lambda}]} |\kappa^{*\prime}(\lambda)|, \quad M_2:=  \max_{\lambda \in [\underline{\lambda}, \overline{\lambda}]} |\kappa^{*\prime\prime}(\lambda)|.
\]

Having defined the necessary objects, we can now proceed to proving Lemma \ref{lem:new}. Consider first the case where the threshold rule allocates the good to agents with merit at the threshold: $\hat x(\eta^*)=0$ (the case of $\hat x(\eta^*)=1$ will be analogous). Take the following ordeal rule, where $\psi,\zeta$ are positive constants:
\begin{equation}\label{eq:lambda}
        \tilde q(\kappa,\lambda) =
        \begin{cases}
                \psi - \zeta\cdot \overbar \lambda + \zeta\cdot \lambda  - \kappa^{*\prime}(\lambda) & \text{if $\tilde \eta(\kappa,\lambda)> \eta^*$},\\
                \psi - \zeta\cdot \overbar \lambda +\zeta \cdot \lambda & \text{otherwise}.
        \end{cases}
\end{equation}
Let the indirect utility function be:
\begin{equation}\label{eq:uform2}
        V(\kappa,\lambda) = \max[0, \ \kappa - \kappa^*(\lambda)] + \zeta\cdot \frac{\lambda^2}{2} + \lambda \cdot(\psi - \zeta\cdot \overbar \lambda).
\end{equation}
I will show this ordeal rule and indirect utility function satisfy the conditions of Lemma \ref{lem:new} for some choice of constants $\psi,\zeta$. First, notice that since $\kappa^*(\lambda)$ is continuous, $V(\kappa,\lambda)$ is continuous too. Thus, $V(\kappa,\lambda)$ is bounded on the extended type space. Notice also that setting $\psi \geq M_1+\zeta \cdot | \overline \lambda - \underline \lambda|$ guarantees that the ordeal rule $\tilde q(\kappa,\lambda)$ is positive on the extended type space.

It remains to show we can select constants to make $V(\kappa,\lambda)$ convex and such that the allocation and ordeal rule pair belongs to its subdifferential everywhere. To that end, I establish the monotonicity of the gradient $\nabla V(\kappa,\lambda)$, wherever it exists. Notice also that this gradient is equal to $[\hat x(\tilde \eta(\kappa,\lambda)) , \tilde q(\kappa,\lambda)]$.

\begin{fact}\label{fact:gradmono}
There exists $\zeta$ such that for any $(\kappa_1,\lambda_1), (\kappa_2,\lambda_2)\in [\underline{\kappa}, \overline{\kappa}] \times [\underline{\lambda}, \overline{\lambda}]$ where the gradient $\nabla V$ exists, we have:
        \begin{equation}\label{eq:convnondiffcond}
        \Delta \kappa \cdot \Delta \tilde x + \Delta \lambda \cdot \Delta \tilde q \geq 0,
        \end{equation}
        where
        $\Delta \kappa := \kappa_2 - \kappa_1, \ \Delta \lambda := \lambda_2 - \lambda_1, \ \Delta \tilde x := \tilde x(\kappa_2,\lambda_2) - \tilde x(\kappa_1,\lambda_1)$ and $\Delta \tilde q := \tilde q(\kappa_2,\lambda_2) -\tilde q(\kappa_1,\lambda_1)$.
\end{fact}
\begin{proof}
Consider three cases. 

        \textbf{Case 1:} $\Delta \lambda=0$. This case follows instantly from the fact that $\hat x$ is weakly increasing and $\tilde \eta(\kappa,\lambda)$ is increasing in $\kappa$.
        
        \textbf{Case 2:} $\Delta \lambda,\Delta \tilde x \neq 0$. Assume without loss that $\tilde x(\kappa_2,\lambda_2) = 1$. Then it has to be that $\tilde \eta(\kappa_2,\lambda_2) > \eta^* \geq \tilde \eta(\kappa_1,\lambda_1)$, so inequality \eqref{eq:convnondiffcond} becomes:
        \[
        \Delta \kappa + \Delta \lambda [\Delta \lambda \cdot \zeta - \kappa^{*\prime}(\lambda_2)] \geq 0.      
        \]
        Equivalently:
        \begin{equation}\label{eq:proof494578}
        \zeta \geq \frac{1}{(\Delta \lambda)^2}  [\Delta \lambda \cdot \kappa^{*\prime}(\lambda_2) - \Delta \kappa].
        \end{equation}                
Notice that to prove \eqref{eq:convnondiffcond} holds for all such $\Delta \kappa,\Delta \lambda$, it suffices to uniformly bound the RHS of \eqref{eq:proof494578} across them. If such a uniform bound exists, we can then simply choose $\zeta$ above it.
        
Since $\tilde \eta(\kappa_2,\lambda_2) > \eta^* \geq \tilde \eta(\kappa_1,\lambda_1)$, it has to be that $\kappa^{*}(\lambda_1)\geq \kappa_1$ and $\kappa^{*}(\lambda_2)\leq \kappa_2$ (Figure \ref{fig:bound1}). This gives the following inequality:
        \begin{equation}\label{eq:bound12}
                \Delta \kappa = \kappa_2 - \kappa_1 \geq \kappa^{*}(\lambda_2) - \kappa^{*}(\lambda_1).
        \end{equation}
        \begin{figure}[!htb]
                \centering
                \begin{tikzpicture}[scale=1] 
                    \draw[->,thick] (0,0) -- (4,0);
                    \draw[->,thick] (0,0) -- (0,4);
                    
                    \node[below left] at (4,0) {$\lambda$};
                    \node[below left] at (0,4) {$\kappa$};
                
                    \node[blue, right] at (3,0.8) {$\kappa^{*}(\lambda)$};
                    
                    \draw[blue, ultra thick] (0.3,3.7)  .. controls (0.9,3.4) and (2.2,0.7) .. (3.7,0.3);
        
                    \draw[red, ultra thick] (2.8,0.81) -- (0.8,3.2);
                    \draw[brown, ultra thick] (0.8,2.65) -- (2.8,1.35);
        
                    \draw[black, thick, dashed] (2.8,0) -- (2.8,1.35);
                    \draw[black, thick, dashed] (0.8,0) -- (0.8,3.2);
        
                    \draw[black, thick, dashed] (0,2.65) -- (0.8,2.65);
                    \draw[black, thick, dashed] (0,3.2) -- (0.8,3.2);

                    \draw[black, thick, dashed] (0,1.35) -- (2.8,1.35);
                    \draw[black, thick, dashed] (0,0.81) -- (2.8,0.81);

                    \fill[black, thick] (0.8,2.65) circle (2pt);
                    \fill[black, thick] (0.8,3.2) circle (2pt);
        
                    \fill[black, thick] (2.8,1.35) circle (2pt);
                    \fill[black, thick] (2.8,0.81) circle (2pt);
        
                    \node[below] at (2.8,0) {$\lambda_2$};
                    \node[below] at (0.8,0) {$\lambda_1$};
        
                    \node[left] at (0,1.35) {$\kappa_2$};
                    \node[left] at (0,0.8) {$\kappa^{*}(\lambda_2)$};
                    \node[left] at (0,3.2) {$\kappa^{*}(\lambda_1)$};
                    \node[left] at (0,2.65) {$\kappa_1$};
        
                \end{tikzpicture}
                \vspace{-10pt} 
                \caption{}
                \label{fig:bound1}
        \end{figure} 

        We can now use \eqref{eq:bound12} to bound the RHS of \eqref{eq:proof494578} from above:
        \begin{align*}
                \frac{1}{(\Delta \lambda)^2}  [\Delta \lambda \cdot \kappa^{*\prime}(\lambda_2) - \Delta \kappa]
                &\leq 
                \frac{1}{(\Delta \lambda)^2}  [\Delta \lambda \cdot \kappa^{*\prime}(\lambda_2) - (\kappa^{*}(\lambda_2) - \kappa^{*}(\lambda_1))]\\
        &=\frac{1}{(\Delta \lambda)^2}  \left[\Delta \lambda \cdot \kappa^{*\prime}(\lambda_2) - \int_{\lambda_1}^{\lambda_2} \kappa^{*\prime}(\tau) \ d\tau \right]\\
        &=\frac{1}{(\Delta \lambda)^2}  \left[\Delta \lambda \cdot \kappa^{*\prime}(\lambda_2) - 
        \int_{\lambda_1}^{\lambda_2}\left(
        \kappa^{*\prime}(\lambda_2) - \int_{\tau}^{\lambda_2} \kappa^{* \prime \prime}(\nu)  d\nu \right)
         \ d\tau \right]\\
         &= \frac{1}{(\Delta \lambda)^2} \int_{\lambda_1}^{\lambda_2} \int_{\tau}^{\lambda_2} \kappa^{* \prime \prime}(\nu) d\nu d\tau\\
         &\leq  \frac{1}{(\Delta \lambda)^2}  \frac{(\Delta \lambda)^2}{2}M_2 = \frac{M_2}{2}.
        \end{align*}
        Thus, setting $\zeta > M_2/2$ ensures that \eqref{eq:convnondiffcond} holds wherever $\nabla V(\kappa,\lambda)$ exists.
        
        \textbf{Case 3:} $\Delta \lambda \neq 0, \Delta \tilde x=0$. Since $\zeta \geq 0$, the case where $\tilde x(\kappa_2,\lambda_2)=\tilde x(\kappa_1,\lambda_1)=0$ is trivial. If $\tilde x(\kappa_2,\lambda_2)=\tilde x(\kappa_1,\lambda_1)=1$, it has to be that $\tilde \eta(\kappa_2,\lambda_2), \tilde \eta(\kappa_1,\lambda_1)>\eta^*$, and so \eqref{eq:convnondiffcond} becomes:
        \[
        \Delta \lambda [\Delta \lambda \cdot \zeta - (\kappa^{*\prime}(\lambda_2) -\kappa^{*\prime}(\lambda_1) )]\geq 0.
        \]
        Equivalently:
        \begin{align*}
        \zeta &\geq \frac{1}{(\Delta \lambda)^2}\Delta \lambda [\kappa^{*\prime}(\lambda_2) -\kappa^{*\prime}(\lambda_1) ]\\
        & = \frac{1}{(\Delta \lambda)^2}\Delta \lambda \int_{\lambda_1}^{\lambda_2} \kappa^{*\prime\prime}(\tau) d\tau.
        \end{align*}
        We can now uniformly bound the RHS from above as follows:
        \[
        \frac{1}{(\Delta \lambda)^2}\Delta \lambda \int_{\lambda_1}^{\lambda_2} \kappa^{*\prime\prime}(\tau) d\tau 
        \leq 
        \frac{1}{(\Delta \lambda)^2}\Delta \lambda \int_{\lambda_1}^{\lambda_2} M_2 d\tau = M_2.
        \]
        Therefore, $V(\kappa,\lambda)$ is convex whenever $\zeta > M_2$. 
\end{proof}

Now, notice that $\nabla V(\kappa,\lambda)$ exists almost everywhere on $[\underline \kappa,\overbar \kappa]\times[\underline \lambda,\overbar\lambda]$, the only exception being the points on the curve $\kappa^*(\lambda)$. Fix a $\zeta$ described in Fact \ref{fact:gradmono}. I will now show that, whenever it exists, the gradient of $V(\kappa,\lambda)$ is also its subgradient. 

Take any $(\kappa_2,\lambda_2)$ in the extended type space and any $(\kappa_1,\lambda_1)$ in the extended type space such that $\nabla V(\kappa_1,\lambda_1)$ exists. Consider the following parametrization of $V(\kappa,\lambda)$:
\[
W(t) := V(\kappa_1+ t(\kappa_2-\kappa_1) ,\lambda_1 + t(\lambda_2-\lambda_1)). 
\]
Then $W(0)= V(\kappa_1,\lambda_1)$ and $W(1)= V(\kappa_2,\lambda_2)$. Since $V(\kappa,\lambda)$ is absolutely continuous, $W'(t)$ exists almost everywhere and is given by:
\[
W'(t) = \nabla V(\kappa_1 + t\Delta \kappa, \lambda_1+t\Delta \lambda )\cdot [\Delta \kappa,\Delta \lambda],
\]
where $\Delta \kappa := \kappa_2 - \kappa_1$ and $\Delta \lambda := \lambda_2 - \lambda_1$. Now, observe that:
\[
t(W'(t) - W'(0)) = [\nabla V(\kappa_1 + t\Delta \kappa, \lambda_1+t\Delta \lambda ) - \nabla V(\kappa_1, \lambda_1)]\cdot[(\kappa_1 + t\Delta\kappa)-\kappa_1,(\lambda_1 + t\Delta \lambda)-\lambda_1].
\]
Since $\nabla V(\kappa,\lambda) = [\tilde x(\kappa,\lambda), \tilde q(\kappa,\lambda)]$, Fact \ref{fact:gradmono} implies that $W'(t) - W'(0)\geq 0$.

Now, since $V(\kappa,\lambda)$ is absolutely continuous, we have:
\begin{align*}
V(\kappa_2,\lambda_2) = W(1) &= W(0) +\int_0^1 W'(t) dt\\
&\geq W(0) +\int_0^1 W'(0) dt\\
&\geq W(0) +\nabla V(\kappa_1, \lambda_1)\cdot [\Delta \kappa,\Delta \lambda]\\
&=V(\kappa_1,\lambda_1) + [\tilde x(\kappa_1,\lambda_1), \tilde q(\kappa_1,\lambda_1)]\cdot [\Delta \kappa,\Delta \lambda].
\end{align*}
Thus, $[\tilde x(\kappa_1,\lambda_1), \tilde q(\kappa_1,\lambda_1)]$ is indeed the subgradient of $V(\kappa,\lambda)$ almost everywhere on $[\underline \kappa,\overbar \kappa]\times [\underline \lambda,\overbar\lambda]$, the exception being the curve $\kappa^*(\lambda)$. However, by continuity of $V(\kappa,\lambda)$, $[0,\psi - \zeta\cdot \overbar \lambda +\zeta \cdot \lambda]$ also belongs to the subdifferential of $V(\kappa,\lambda)$ there, and so $[\tilde x(\kappa_1,\lambda_1), \tilde q(\kappa_1,\lambda_1)]$ belongs to the subdifferential of $V(\kappa,\lambda)$ everywhere on $[\underline \kappa,\overbar\kappa]\times [\underline \lambda,\overbar \lambda]$. Finally, since $V(\kappa,\lambda)$ has a monotone subgradient at every point of its domain, it is convex.

The proof for the case where $\hat x(\eta^*)=1$ is analogous, except the ordeal rule $\tilde q(\kappa,\lambda)$ is set to $\psi - \zeta\cdot \overbar \lambda +\zeta \cdot \lambda - \kappa^{*\prime}(\lambda)$ along the curve $\kappa^*(\lambda)$.

\subsection{Proof of Proposition \ref{prop:relax}}

The proof consists of two steps. First, I establish a lower bound on the global equity violations of all non-constant allocation rules implemented with only payments. To that end, I consider a type around which the allocation is not locally constant and show that all types with the same allocation as this one must lie below some differentiable curve. I establish that these curves are downwards-sloping when only ordeals are used to screen and upwards-sloping when only payments are used. I then derive expressions for their slopes and use them to bound from below the global equity violations of all allocation rules implemented with only payments.

I then assume that iso-merit curves are sufficiently flat and construct an allocation rule implemented with only ordeals whose global equity violation is below the aforementioned lower bound.

I now proceed with the first step. To that end, consider some non-constant allocation rule $x$ implemented only with ordeals (the case where we screen only with payments is analogous). I first prove a few facts about such allocation rules.

\begin{fact}[Monotonicity]\label{fact:qdom}
If $\alpha^a \geq \alpha^b$ and $\beta^a \geq \beta^b$ with at least one inequality holding strictly, then type $(\alpha^a,\beta^a)$ gets a weakly higher allocation than type $(\alpha^b,\beta^b)$: $x^a \geq x^b$.
\end{fact}
\begin{proof}
        Suppose that $x^a < x^b$. Then $q^a<q^b$ or else both types would strictly prefer $(x^b,q^b)$. By revealed preference:
        \[
        v(\beta^b,x^b) - z(\alpha^b,q^b) \geq v(\beta^b,x^a) - z(\alpha^b,q^a)  
        \quad
        \Longrightarrow
        \quad
        v(\beta^b,x^b) - v(\beta^b,x^a) \geq z(\alpha^b,q^b) - z(\alpha^b,q^a).
        \]
        By strictly increasing differences we get:
        \[
        v(\beta^a,x^b) - v(\beta^a,x^a) > z(\alpha^a,q^b) - z(\alpha^a,q^a)
        \quad
        \Longrightarrow
        \quad
        v(\beta^a,x^b) - z(\alpha^a,q^b) > v(\beta^a,x^a) - z(\alpha^a,q^a).
        \]
        That is, $(\alpha^a,\beta^a)$ prefers $(x^b,q^b)$ to $(x^a,q^a)$; contradiction.    
\end{proof}
In particular, Fact \ref{fact:qdom} tells us that $x_\delta^a$ is weakly increasing in $\delta$ for every $(\alpha^a,\beta^a)$. Moreover, whenever $x^a >x^b$ it has to be that $q^a>q^b$ or else $(\alpha^b,\beta^b)$ would prefer $(x^a,q^a)$ to her allocation. Therefore, $q^a_\delta$ is also weakly increasing in $\delta$ for every $(\alpha^a,\beta^a)$.

The following fact proves the existence of a sequence that lets us apply the Generalized FOC (Lemma \ref{lem:MRS}) to some type $(\alpha^a,\beta^a)$.

\begin{fact}\label{lem:notconst}
There exists $(\alpha^a,\beta^a)$ such that either $a)$ for every $\delta>0$, $x^a_\delta>x^a$, or $b)$ for every $\delta<0$, $x^a_\delta<x^a$.
\end{fact}
\begin{proof}
Suppose otherwise; then for every $(\alpha^a,\beta^a) \in \Theta$ there exists some $\epsilon>0$ such that $x^a_\delta=x^a$ for $\delta \in [-\epsilon,\epsilon]$. Take any $(\alpha^b,\beta^b)\in [\alpha^a-\epsilon, \alpha^a + \epsilon] \times [\beta^a-\epsilon, \beta^a + \epsilon]$ and notice that by Fact \ref{fact:qdom} we have $x(\alpha-\epsilon,\beta - \epsilon) \leq x^b \leq x(\alpha+\epsilon,\beta+\epsilon)$. However, $x(\alpha-\epsilon,\beta-\epsilon) = x(\alpha+\epsilon,\beta+\epsilon)= x^a$ so $x(\alpha,\beta)=x^a$ for all $(\alpha,\beta) \in [\alpha^a-\epsilon, \alpha^a + \epsilon] \times [\beta^a-\epsilon, \beta^a + \epsilon]$. Consequently, for any $(\alpha,\beta)\in \Theta$ there exists a neighborhood around it in which the allocation is constant. Now, take any $(\alpha^c,\beta^c), (\alpha^d,\beta^d)\in \Theta$. Since $\Theta$ is connected, there exists a continuous path between them and every point along this path has the same allocation as the points within its neighborhood. Therefore, the allocation has to be constant along the whole path, including the end-points. Since  $(\alpha^c,\beta^c)$ and $(\alpha^d,\beta^d)$ were arbitrary, the allocation is constant; contradiction. 
\end{proof}
Fix $(\alpha^a,\beta^a)$ from Fact \ref{lem:notconst} and assume $a)$ holds (the argument for $b)$ is analogous). Let $x_+^a = \lim_{\delta \to 0^+}x^a_\delta$ and $q_+^a = \lim_{\delta \to 0^+}q^a_\delta$. These one-sided limits exist as $x^a_\delta$ and $q^a_\delta$ are increasing in $\delta$.

\begin{fact}\label{fact:ineqs} If $x^a_\delta$ is right-continuous at $\delta=0$, then for every $(\alpha^b,\beta^b)$ such that $x^b = x^a$ we have:
        \begin{equation}\label{eq:cond1imeqs}
                \frac{v_x(\beta^b,x^a)}{z_q(\alpha^b,q^a)} \leq \frac{v_x(\beta^a,x^a)}{z_q(\alpha^a,q^a)}.
        \end{equation}
If $x^a_\delta$ is not right-continuous at $\delta=0$, then for every $(\alpha^b,\beta^b)$ such that $x^b = x^a$ we have:
\begin{equation}\label{eq:cond2imeqs}
        \frac{v(\beta^b,x^a_+) - v(\beta^b,x^a)}{z(\alpha^b,q^a_+) - z(\alpha^b,q^a)} 
        \leq 
        \frac{v(\beta^a,x^a_+) - v(\beta^a,x^a)}{z(\alpha^a,q^a_+) - z(\alpha^a,q^a)}.
\end{equation}
\end{fact}
\begin{proof} Consider two cases.

\textbf{Case 1:} $x^a_\delta$ is right-continuous at $\delta=0$. Take a decreasing sequence $\{e_i\}_i$ such that $e_i\to 0$. Then $x^a_{e_i} \to x^a$ by right-continuity and, by case $a)$ of Fact \ref{lem:notconst}, $x^a<x^a_{e_i}$ for all $i$. Since $x^a_\delta$ is right-continuous at $\delta =0$, there exists a subsequence $\{f_i\}_i$ of $\{e_i \}_i$ such that $x_{f_i}^a$ is decreasing strictly. We can therefore apply the Generalized FOC (Lemma \ref{lem:MRS}), which completes the proof.

\textbf{Case 2:} $x^a_\delta$ is not right-continuous at $\delta=0$. Then $x_+^a > x^a$. Since indirect utility has to be continuous at $(\alpha^a,\beta^a)$, it has to be that $q^a_+>q^a$ and:
\begin{equation}\label{eq:eee}
        v(\beta^a,  x^a) - z(\alpha^a, q^a)= v(\beta^a, x^a_+) - z(\alpha^a, q^a_+)
        \quad
        \Longrightarrow
        \quad
        \frac{v(\beta^a, x^a_+) -v(\beta^a,  x^a) }{z(\alpha^a, q^a_+) - z(\alpha^a, q^a)} =1.
\end{equation}
Take a decreasing sequence $\{\delta_i\}_i$ such that $\delta_i\to 0$. By the fact that $x^b=x^a$ and revealed preference, for every $i$ we have:
\[
v(\beta^b,x^a)-z(\alpha^b,q^a) \geq v(\beta^b,x^a_{\delta_i})-z(\alpha^b,q^a_{\delta_i})
\quad
\Longrightarrow
\quad 
\frac{v(\beta^b,x^a_{\delta_i}) - v(\beta^b,x^a)}{z(\alpha^b,q^a_{\delta_i}) - z(\alpha^b,q^a)}\leq 1.
\]
Taking $i \to \infty$ gives:
\begin{equation}\label{eq:fff}
\frac{v(\beta^b,x^a_+) - v(\beta^b,x^a)}{z(\alpha^b,q^a_+) - z(\alpha^b,q^a)} \leq 1.   
\end{equation}
Combining the latter equality in \eqref{eq:eee} with \eqref{eq:fff} completes the proof.
\end{proof}

The LHSs of \eqref{eq:cond1imeqs} or \eqref{eq:cond2imeqs} are strictly increasing in $\alpha^b$ and $\beta^b$. Thus, to satisfy conditions \eqref{eq:cond1imeqs} or \eqref{eq:cond2imeqs} of Fact \ref{fact:ineqs}, type $(\alpha^b,\beta^b)$ must lie weakly below some differentiable curve. If $x^a_\delta$ is right-continuous at $\delta=0$, this curve is traced out by $(\alpha,\beta)$ for which:
\begin{equation}\label{eq:isomrs}
        \frac{v_x(\beta,x^a)}{z_q(\alpha,q^a)}=\frac{v_x(\beta^a,x^a)}{z_q(\alpha^a,q^a)}. 
\end{equation}
I will call this curve the \emph{iso-MRS curve at $(\alpha^a,\beta^a)$}. If $x^a_\delta$ is not right-continuous at $\delta=0$, this curve is given by:
\begin{equation}\label{eq:isodifference}
        \frac{v(\beta^b,x_+^a) - v(\beta^b,x^a)}{z(\alpha^b,q^a_+) - z(\alpha^b,q^a)}=\frac{v(\beta^a,x_+^a) - v(\beta^a,x^a)}{z(\alpha^a,q_+^a) - z(\alpha^a,q^a)}.
\end{equation}
I will call it the \emph{iso-difference curve at $(\alpha^a,\beta^a)$}.

Recall that $D(\alpha^a,\beta^a)$ from Definition \ref{def:eqviol} is the set of directions in which the allocation is locally constant at type $(\alpha^a,\beta^a)$. The above observation then tells us that for $(\alpha^a,\beta^a)$ with merit $\eta^a$ there is only one candidate element of $D(\alpha^a,\beta^a)$: depending on the right-continuity of $x^a_\delta$ at 0, it is either the direction with the slope of the iso-$MRS$ curve at $(\alpha^a,\beta^a)$ or the slope of the iso-difference curve at $(\alpha^a,\beta^a)$. These slopes are given by:
\begin{equation}\label{eq:slopeq1}
        s_{MRS}^q(\alpha^a,\beta^a) := \frac{v_x(\beta^a,x^a)}{z_q(\alpha^a,q^a)}\frac{z_{\alpha q}(\alpha^a,q^a)}{v_{\beta x}(\beta^a,x^a)},
\end{equation}
\begin{equation}\label{eq:slopeq2}
        s_{diff}^q(\alpha^a,\beta^a):= \frac{v(\beta^a,x_+^a) - v(\beta^a,x^a)}{z(\alpha^a,q_+^a) - z(\alpha^a,q^a)}
        \frac
        {z_\alpha(\alpha^a,q_+^a) - z_\alpha(\alpha^a,q^a)}
        {v_\beta(\beta^a,x_+^a) - v_\beta(\beta^a,x^a)}.
\end{equation}

Note $s_{MRS}^q(\alpha,\beta), \ s_{diff}^q(\alpha,\beta)<0$ for every $(\alpha,\beta)\in\Theta$.

Now, consider the case where we screen only with payments. Analogous reasoning pins down the only candidate elements of $D(\alpha^a,\beta^a)$ for $(\alpha^a,\beta^a)$ such that $x^a_\delta$ is locally non-constant at $\delta=0$. Then, however, $s_{MRS}^p(\alpha,\beta), \ s_{diff}^p(\alpha,\beta)>0$ for every $(\alpha,\beta)\in \Theta$.

Notice also that the slope of the iso-merit curve at $(\alpha,\beta)$ equals to $-\eta_\alpha(\alpha,\beta)/\eta_\beta(\alpha,\beta)$ and is negative for any $(\alpha,\beta)\in\Theta$. We therefore have the following lower bound on global equity violations for any allocation rule $x_p$ implemented with payments only:
\begin{equation}\label{eq:pbound}
                L(x_p) > 
                \Big| \frac{\pi}{2} - \angle\left(-\frac{\eta_\alpha(\alpha,\beta)}{\eta_\beta(\alpha,\beta)}\right) \Big|.        
\end{equation}

Having established this uniform lower bound, I proceed to the second step of the proof. I show that if iso-merit curves are sufficiently flat everywhere, I can construct an allocation rule implemented with ordeals only whose global equity violation is strictly below this lower bound.

To that end, I first establish a uniform bound on the slopes $s_{MRS}^q(\alpha^a,\beta^a)$ and $s_{diff}^q(\alpha^a,\beta^a)$ that holds across types and all allocation rules $x_q$ implemented with only ordeals. Recall that $v,z$ as well as their partials and cross-partials are continuous. Moreover, these partials and cross-partials are non-zero everywhere. Thus, the absolute values of $v,v_{\beta x}$ and $z_q,z_{\alpha q}$ can be uniformly bounded away from zero and infinity on the whole type space $\Theta$ and for all $x\in[0,1]$, $q\in [0,\overbar q]$. Analogous bounds hold for the following terms whenever $x^b \neq x^b_+$ (which implies $q^b \neq q^b_+$): 
\[
\frac{v(\beta^b,x_+^b) - v(\beta^b,x^b)}{x^b_+-x^b},
\ \ \frac{z(\alpha^b,q_+^b) - z(\alpha^b,q^b)}{q^b_+-q^b},
\ \  \frac{v_\beta(\beta^b,x_+^b) - v_{\beta}(\beta^b,x^b)}{x^b_+-x^b},
\ \  \frac{z_\alpha(\alpha^b,q_+^b) - z_\alpha(\alpha^b,q^b)}{q^b_+-q^b}.
\]
Therefore, by \eqref{eq:slopeq1} and \eqref{eq:slopeq2}, there exists $M <0$ such that for any $(\alpha,\beta)\in\Theta$ and $x_q$: 
\[
s_{MRS}^q(\alpha,\beta), \ s_{diff}^q(\alpha,\beta) < M.
\]
That is, all iso-$MRS$ and iso-difference curves for all allocation rules $x_q$ implemented with only ordeals are steeper than $M$. Now, assume that all iso-merit curves are flatter than $M$:
\[
\text{for all }(\alpha,\beta)\in \Theta, \quad -\frac{\eta_\alpha(\alpha,\beta)}{\eta_\beta(\alpha,\beta)} > M.
\]
I will construct an allocation rule implemented with only ordeals whose equity violation is below the bound in \eqref{eq:pbound}. Fix any $(\alpha^b,\beta^b)\in \Theta$. Then, for $x^b>0$ small enough, $q^b$ defined by
\[
        v(\beta^b,x^b) = z(\alpha^b,q^b),
\]
is sufficiently close to $0$ that $q^b \in [0,\overbar q]$. The following allocation then satisfies \eqref{eq:IR} and \eqref{eq:IC}: $(\alpha^b,\beta^b)$ and all types weakly above the iso-difference curve at $(\alpha^b,\beta^b)$ take $(x^b,q^b)$; all types strictly below it take $(0,0)$. 

Take any $(\alpha^c,\beta^c)$ strictly below or strictly above the iso-difference curve at $(\alpha^b,\beta^b)$. There is a neighborhood around $(\alpha^c,\beta^c)$ in which the allocation is constant, and so it is also constant along the slope of the iso-merit curve at  $(\alpha^c,\beta^c)$, giving a local equity violation $l(\alpha^c,\beta^c)=0$ there. Now, take any $(\alpha^c,\beta^c)$ on the iso-differencecurve at $(\alpha^b,\beta^b)$. The allocation is constant along this curve, so $s_{diff}^q(\alpha^c,\beta^c)$ is the slope of some direction in $D(\alpha^c,\beta^c)$. Since the iso-merit curve is flatter than $s_{diff}(\alpha^c,\beta^c)$ for any such $(\alpha^c,\beta^c)$, we get:
\[
L(x_q) \leq \Big|\angle(s_{diff}^q(\alpha^c,\beta^c)) - \angle\left(-\frac{\eta_\alpha(\alpha^c,\beta^c)}{\eta_\beta(\alpha^c,\beta^c)}\right) \Big| <
\Big| \frac{\pi}{2} - \angle\left(-\frac{\eta_\alpha(\alpha^c,\beta^c)}{\eta_\beta(\alpha^c,\beta^c)}\right) \Big|
<L(x_p). 
\]


\end{document}